\documentclass{article}
\usepackage{graphicx}
\usepackage{balance}  
\usepackage{url}
\usepackage{epstopdf}

\usepackage{algpseudocode}
\usepackage{algorithmicx}
\usepackage{algorithm}
\usepackage{subfig}
\usepackage{authblk}
\usepackage{amsmath}
\usepackage{hyperref}
\usepackage{amsthm}
\usepackage{setspace}

\begin{document}
\begin{spacing}{1.0}

\title{Mining CFD Rules on Big Data}



%
%
%
%


\author[a]{Hongzhi Wang}
\author[a]{Mingda Li}
\author[a]{Jiawei Zhao}

\affil[a]{Harbin Institute of Technology,wangzh@hit.edu.cn}
       
%


\maketitle

\begin{abstract}
	
Current conditional functional dependencies (CFDs) discovery algorithms always need a well-prepared training data set. This makes them difficult to be applied on large datasets which are always in low-quality. To handle the volume issue of big data, we develop the sampling algorithms to obtain a small representative training set. For the low-quality issue of big data, we then design the fault-tolerant rule discovery algorithm and the conflict resolution algorithm. We also propose parameter selection strategy for CFD discovery algorithm to ensure its effectiveness. Experimental results demonstrate that our method could discover effective CFD rules on billion-tuple data within reasonable time.
\end{abstract}

\section{Introduction}
	Nowadays, with the accumulation of data, the size of database becomes larger and larger. At the same time, due to the difficulty in manually maintenance and the variance type of data sources, big data contains quality problems with higher possibility, making it difficult to be used. Therefore, big data cleaning techniques are crucial for its effective usage.


	CFDs~\footnote{Fan2008Conditional} are powerful tools for data cleaning. They can find the hidden relation among items. Such relation can help us to find dirty tuples, which can be modified accordingly. The functional dependencies (FDs) can be considered as special forms of CFDs. High-quality rules are the core of effective data cleaning systems with CFDs.

High-quality CFD discovery on big data brings two challenges. On one hand, the volume of big data require high-efficiency and high-scalability discovery algorithm, whose complexity is linear or sub-linear to the data size. On the other hand, big data may involve more data quality problems. Thus, a clean training set can hardly be prepared for CFD discovery. Thus, fault-tolerant CFD discovery approach is in demand.

Due to its importance, some CFD discovery algorithms have been proposed. However, none of them could trakle the challenges.
 Most of existing methods such as \cite{Chen2009Analyses} discover high-quality rules with data mining algorithms on a small but clean data set efficiently. However, these approaches are unsuitable for big data cleaning due to the lack of representative data set. Other methods for discovering CFDs on dirty data set such as \cite{Bollob1998Modern} need many passes over the data set to find approximate CFDs, but for a big data set which cannot be loaded into the main memory, it can hardly work.

Therefore, it is necessary to develop a scalable method to mine high-quality rules from big data with size larger than the main memory. To achieve this goal, we design a scalable and systemic algorithm. We sample from the big data firstly to obtain effective training set within one pass of scan. Then, we discover CFDs based on sampling results. The reason for sampling before discovery is that, without sampling, we have to scan the big data set for many times to mine patterns and calculate support for finding CFDs. This is time-consuming, especially when data set is larger than memory. Another purpose for sampling is to filter dirty items and keep clean ones. Following example shows the necessary of sampling.
	
	\textbf{Example 1:} Table \ref{tab: example 1} is changed from the example in \cite{Golab2008On}. It is about a customer with the basic information (country code (CC), area code (AC), phone number (PN)), name (NM), and address (street (STR ), city (CT), zip code (ZIP )).

\begin{table*}
	\centering
	\caption{Example 1}
	\label{tab: example 1}
	\begin{tabular}{|l|l|l|l|l|l|l|l|}
		\hline
		& CC & AC & PN & NM & STR & CT & ZIP \\
		\hline
		$t_1$ & 01 & 108 & 11080176 & Ian & Three Ave. & MH & 2221 \\
		\hline
		$t_2$ & 01 & 108 & 11080176 & Jack & Tree Ave. & MH & 2221 \\
		\hline
		$t_3$ & 01 & 112 & 11120101 & Joe & High St. & NYC & 02ED1 \\
		\hline
		$t_4$ & 01 & 108 & 11120101 & Jim & Elm Str. & MH & 2221 \\
		\hline
		$t_5$ & 40 & 1069 & 41690101 & Ben & High St. & EDI & 02ED1 \\
		\hline
		$t_6$ & 40 & 1069 & 41690177 & Ian & High St. & EDI & 02ED1 \\
		\hline
		$t_7$ & 40 & 108 & 41690177 & Ian & Port PI & MH & 02WB2 \\
		\hline
		$t_8$ & 01 & 1069 & 11120101 & Sean & 3rd Str. & UN & 2233 \\
		\hline
		$t_9$ & 4731 & 108 & 233323 & Steve & Low St. & SYD & XXXX \\
		\hline
		$t_{10}$ & 4731 & XXXX & 3456123267 & Steve & Low St. & LON & 2112E \\
		\hline
		$t_{11}$ & 8E11 & 979797 & 678345 & Laola & 4th St. & MH & 322233 \\
		\hline
	\end{tabular}
\end{table*}

From the table, we can find the traditional FD set $f_{1-2}$:
\begin{align}
	&f_1:[{\rm CC},{\rm AC}]\rightarrow {\rm CT}\\
	&f_2:[{\rm CC},{\rm AC},{\rm PN}]\rightarrow {\rm STR}
\end{align}

\noindent and the CFD set $\beta_{1-5}$:
\begin{align*}
&\beta_1:([{\rm CC},{\rm ZIP}]\rightarrow {\rm STR},(40,\underline{\quad}\parallel\underline{\quad}))\\
&\beta_2:([{\rm CC},{\rm AC}]\rightarrow {\rm CT},(01,112\parallel{\rm NYC}))\\
&\beta_3:([{\rm CC},{\rm AC}]\rightarrow {\rm CT},(01,108\parallel{\rm MH}))\\
&\beta_4:([{\rm CC},{\rm AC}]\rightarrow {\rm CT},(40,1069\parallel{\rm EDI}))\\
&\beta_5:([{\rm CC},{\rm NM}]\rightarrow {\rm STR},(4731,{\rm Steve}\parallel{\rm Low St.}))
\end{align*}

However, if the data set is about the customers in America, then the dirty items $t_9$ and $t_{10}$ with CT-SYD and LON which are not in the US will have little similar items. With the sampling method to find representative samples, we need to ignore them. Moreover, we neglect $t_{11}$, because we cannot find items holding more than two attributes same with the attributes of $t_{11}$, which help us little to find a CFD showing the hidden relation among most of items. Therefore, the following rule set $\varphi_{1-2}$ could be discovered from data set without $t_9$,$t_{10}$,$t_{11}$.
\begin{align*}
&\varphi_1:([{\rm CC},{\rm ZIP}]\rightarrow {\rm STR},(40,\underline{\quad}\parallel\underline{\quad}))\\
&\varphi_2:([{\rm AC}]\rightarrow {\rm CT},(\underline{\quad}\parallel\underline{\quad}))
\end{align*}

If we clean the data set based on the rule set $\beta_{1-5}$, the $t_9$ and $t_{10}$ will conform to $\beta_5$ and be treated as clean items. However, with the new rule set, they will become dirty. Meanwhile, since the attributes of the two new rules are less, we do not need to compare for many times, which is time-consuming for big data. Therefore, the data cleaning with only two rules $\varphi_1$ and $\varphi_2$ is more efficient than that with 5 rules  $\beta_{1-5}$.

From Example 1, we can find the selection of training set is important. Meanwhile, for big data, it is only possible to use a small set of items for rule discovery. Thus, it is significant to select a representative training set from big data. To the big data set with size larger than memory, we attempt to accomplish sampling in one pass as the sampling method for estimating the confidence of CFDs in \cite{Cormode2009Estimating}.

In summary, the developed rule discovery method suitable for big data with size larger than memory requires following features which the existing methods do not have.

\begin{enumerate}
	\item[(1)] A small but representative training set should be selected in one-pass scanning of the data.
	\item[(2)] The method to discover rules from items should tolerate the wrong records in the training set.
	\item[(3)] Due to the tradeoff between effectiveness and efficiency, a mechanism tuning the parameter according to the need of applications should be provided.
\end{enumerate}


Motivated by this, we propose a method for discovering a high-quality CFD set. Such approach could tolerate data quality problems in the data set and meet various requirements from users for a data set with size larger than memory. The contributions of this paper are summarized as follows.

\begin{enumerate}
	\item[(1)] We design BRRSC, a sampling method to obtain a proper training set from CFD discovery within once scanning of data. According to the theoretical analysis and experiments, BRRSC is a sub-linear algorithm suitable for big data.

	\item[(2)] We propose an algorithm DFCFD that could tolerate error data to discover CFDs by our proposed method. DFCFD can be changed according to different data size and the parameter of dirty data set to obtain the best CFD set.

	\item[(3)] To resolve conflicts among the discovered CFD set, we propose a graph-based algorithm with each CFD as a node and the conflict relationship between two CFD as an edge. In this algorithm, the conflict-free CFD set is computed as the maximal weight independent set on the graph.
	
\item[(4)] To meet various requirements for CFD discovery, we design adaptive parameter computation strategy for CFD discovery. We define four dimensions of user requirements. Users are allowed to decide the most important aim in the discovery and set limits for the other three. After that, we propose a multi-objective programming to solve this parameter determination problem.
	
\item[(5)] We verify experimentally the performance and scalability of our algorithm. We compare the time for discovering CFDs and the quality of CFDs with previous methods for different data sizes and parameters. To test the optimality of the parameter selection method, we compare the effectiveness of different choices of parameters using the controlling variable method. Meanwhile, we use the real-life big data to show the effectiveness of our method.
\end{enumerate}


We introduce the preliminary definitions and the framework of our solution in Section~\ref{section: pri} and Section~\ref{section: fra}, respectively. The sampling method is proposed in Section~\ref{section: rrsc}. In Section~\ref{section: bdc}, we develop error-tolerant CFD discovery algorithms and conflicts resolving algorithms. An adaptive parameter selection algorithm is proposed in Section~\ref{section: sel of para}. In Section~\ref{section: exp}, we perform extensive experiments to verify the efficiency and effectiveness of proposed algorithms. Finally, we draw the conclusions in Section~\ref{section: con}.

\section{Priliminary}
\label{section: pri}

In this section, we first review some definitions about CFDs. We then define the problem.

\subsection{Background}

A CFD is a pair $(X\rightarrow A,t_P)$, where $X$ is a set of attributes in the items, $A$ is a single attribute decided by $X$, and $t_P$ is a pattern tuple with attributes in $X$ and $A$. For an attribute $C$ in $X\cup A$, $t_P[C]$ is either a constant or an undetermined variable denoted as ``\underline{\quad}''. We define $X$ and $A$ as LHS and RHS for a CFD, respectively. A pattern tuple ``$\parallel$'' is used to separate the $X$ and $A$ attributes.

We call a CFD as constant CFD if $t_P$ consists of constants only, i.e., $ t_P[A] $ as a constant and $t_P[B]$ as a constant for all $B\in X$. It is called a variable CFD if $t_P[A]$ is ``\underline{\quad}'', and the value of $tP[B]$ depends on that of $ t_P[A] $. As for the general CFDs, they include both of the variable and constant CFDs.

Among the CFD set $\beta_{1-5}$ in Example 1, $\beta_1$ is a variable CFD while the $\beta_{2-5}$ are constant CFDs. In the CFD set $\varphi_{1-2}$, $\varphi_1$ and $\varphi_2$ are both variable CFDs.

When we find CFDs, we should avoid trivial and redundant CFDs to increase efficiency. To achieve this goal, we define the minimal CFDs. A minimal CFD must be a nontrivial, left-reduced CFD firstly.

A CFD $(X\rightarrow A, t_P)$ is trivial when $A\in X$. If a CFD is trivial, it is always correct when the attribute in $ X $ is equal to the same attribute in $ A $. It is always wrong when the equality relationship is not met. Therefore, we only study the nontrivial CFDs in this paper. We call the constant CFD $(X\rightarrow A, (t_P\parallel a))$ a left-reduced CFD if no set of attributes $ Z $ is included in $ X $ to make a new CFD $(Z\rightarrow A, (t_P\parallel a))$. Similarly, we call a variable CFD left-reduced if for any $Z\subset X$, $(Z\rightarrow A, (t_P\parallel a))$ cannot be proved proper, and there is no $t_{P'}[X]$ more general than $t_P[X]$ to make the $(X\rightarrow A, (t_{P'}\parallel a))$ correct. To determine the confidence level of a CFD, we say a tuple support a CFD when it satisfies the condition in $\varphi$.

\begin{figure*}[h]
	\centering
	\includegraphics[width=15cm]{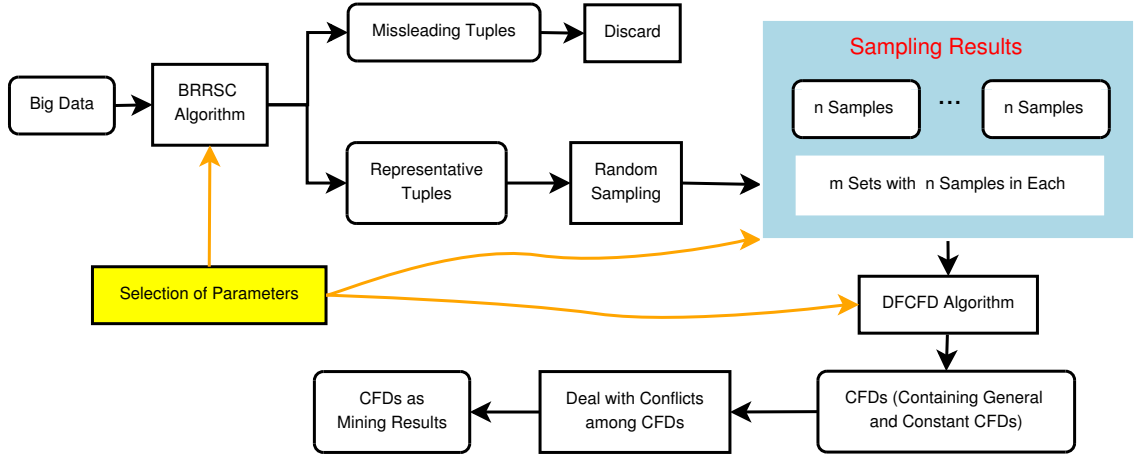}
	\caption{The framework of the whole process}
	\label{fig: framework}
\end{figure*}

\subsection{Problem Definition}

Given a data set which may be quite large, our goal is to find a high-quality CFD set, which contains both constant and variable CFDs. Since for big data, its major part is clean, we regard a CFD set as high-quality when most tuples in big data support it. Meanwhile, a high-quality CFD set should control its CFD number. Thus, we need to discover a CFD set containing minimal number of CFDs with most tuples supporting it. It is difficult to measure the quality when considering both the number of CFDs and supporters. Therefore, in the experiments, we used a standard CFD set discovered on a clean data set. Then, we modified the data set to make it dirty and utilized our method to discover our set of CFDs on it. We give the evaluation of our set of CFDs by comparing them with the standard CFD set.

\section{Framework}
\label{section: fra}

The framework of the proposed method is shown in Figure \ref{fig: framework}. In the working process described by this figure, to obtain a high-quality CFD set from big data, we firstly obtain samples through the algorithm proposed in Section~ \ref{section: one-pass} in one pass scanning. Then, a error-tolerant CFD discovery algorithm in Section \ref{section: dfcfd} is developed to find CFDs from the samples. After that, we establish a weighted undirected graph including CFDs as nodes (Section \ref{section: cal weight}) and add an edge between two CFDs to represent a conflict (Section \ref{section: disc confict}). To deal with the conflicts, we adapt the algorithm in \cite{Bollob1998Modern} for finding a maximal weighted independent set (Section \ref{section: mwid}). Meanwhile, to satisfy different requests from users, we propose a novel method to choose the most suitable parameters for CFD discovery (Section \ref{section: sel of para}). In summary, the proposed system framework is separated into four parts: a sampling algorithm, a error-tolerant dynamical CFD discovery algorithm, a method dealing with conflicts among CFDs and the selection of parameters.

\section{RRSC: Representative and Random Sampling for Cfds}
\label{section: rrsc}

To select a small but representative data set for CFD discovery, we attempt to use sampling method. Although reservoir sampling \cite{Vitter1985Random} could ensure the equal possibility for each tuple to be sampled with unknown size of the whole data, the representativeness of the sample could not be ensured. Thus, inspired by the reservoir sampling, we propose a novel sampling algorithm which calculates the number of the same attributes of samples to decide whether a tuple is suitable. Meanwhile, to ensure our samples represent all kinds of suitable tuples, we choose multiple sets of samples from a big data set. We then find CFDs on each sample set. We then finally synthesize the whole CFD set by modelling all discovered CFDs as a weighted graph and finding the subset with the largest weights. 

We suppose the number of the groups and samples in each group are $n$ and $m$, respectively. In this section, we first propose a multiple-passes scan algorithm in Section~\ref{section: mul-pass}. In this algorithm, we find out $n$ groups of popular items iteratively. This algorithm is divided into two phrases, the first extraction and the $2^{\rm th}$ to $n^{\rm th}$ extractions where $m$ denotes the number of items in each group, since the $2^{\rm th}$ to $n^{\rm th}$ is a process of iteration different from the first extraction. During the $2^{\rm nd}$ to $n^{\rm th}$ extractions where $n$ is the group number, we need to compare samples with both current and original sampling results. However, it is infeasible to scan dataset multiple times for big data. In Section~\ref{section: one-pass}, we will talk about how to perform the iteration in once scan.

\subsection{Multiple-Pass Scan Algorithm}
\label{section: mul-pass}

We start from the discussion of the criteria for selecting or avoiding a tuple included in the sample and then describe the algorithm in Section~\ref{sec:criteria}. The sample is divided into two parts. The first group of $m$ items is obtained primarily as the base, and the $2^{\rm nd}-n^{\rm th}$ groups are sampled in iteratively until all kinds of popular items are sampled. We will discuss these two algorithms in Section \ref{section: frrsc} and Section \ref{section: trrsc}, respectively.

\subsubsection{Tuple Section Criteria}
\label{sec:criteria}

First of all, we should avoid  misleading samples like $ t_9 $, $ t_{10} $ and $ t_{11} $ in Example 1, which are either special or unpopular. A misleading tuple is a tuple with following features:


\begin{enumerate}
	\item[(1)]	If a tuple has at least one incomplete attribute, such as $t_9$ and $t_{10}$, we treat it as a misleading tuple.
	
	\item[(2)]	If we compare the attributes of a tuple $t$ with popular tuples and find that the number of the same attributes is smaller than a threshold $\epsilon$, $t$ is treated as a misleading tuple. $\epsilon$ will be defined according to the method in Section~\ref{section: sel of para}.
\end{enumerate}

Secondly, it is necessary to avoid similar items as to prevent over-fitting. To achieve this goal, we adopt $2^{\rm nd}$ to $n^{\rm th}$ iteration. In the $i^{\rm th}$ sample where $2\leq i\leq n$, we compare it with the samples obtained from the $1^{\rm st}$ to $(i-1)^{\rm th}$ sample. If the number of the same attributes between current item and early results is larger than a threshold, this item is considered too similar for sampling results and given up.

\subsubsection{FRRSC: Sample for the $1^{\rm st}$ Group}
\label{section: frrsc}

\noindent \underline{Algorithm Overview} The $1^{\rm st}$ group is generated by the framework similar as reservoir sampling. The difference is the replacement of sample takes the criteria in Section~\label{sec:criteria} into consideration.

We first include the front $m$ tuples in the sample $S$. For each of the following tuples $t$, we decide whether it is the misleading tuple.
If $t$ is incomplete, we refuse to add it to $S$ directly. Otherwise, we use $1/q$ as the selection probability $t$, where $q$ is the number of tuples in $S$ with sharing more than $\epsilon$ attributes with $t$. Such that too unpopular tuples will be selected in very low probability.

\textbf{Algorithm Description} The pseudo code of the algorithm is shown in Algorithm \ref{algor: FRRSC}.

\begin{algorithm}
	\caption{FRRSC}
	\hspace*{0.02in}{\bf Input:}
	$ \boldsymbol{N}$: the data set, $\boldsymbol{m}$: sample size$\boldsymbol{B'} $ \\
	\hspace*{0.02in}{\bf Output:}
	\begin{algorithmic}
		\State The sample $S$
        \State $k \leftarrow 0$
        \State $i \leftarrow 0$
        \While($k < m$ and $i<|N|$)
            \If{$N_i$ is complete}
                \State $S_k \leftarrow N_i$
                \State $k\leftarrow k+1$
            \EndIf
            \State $i\leftarrow i+1$
        \EndWhile
        \While{$i<|N|$}
            \If{$ N[t] $ is complete}
				\For{$ j=0 $ to $m-1$}
					\If{cmp$(N_i, S_j) \geq \epsilon$}
					   \State $ q\leftarrow q+1 $;$ k \leftarrow $rand$ [1,q] $;
					\EndIf
					\If{$ k \leq m $}
					   \State $ S_j \leftarrow N_i $
                        \State break
					\EndIf
				\EndFor
			\EndIf
            \State $i \leftarrow i+1$
        \EndWhile
        \State return $S$
	\end{algorithmic}
	\hspace*{0.02in}{\bf Note:}
	cmp$ (T[1,i],t) $ shows the number of the same attributes shared by two tuples.
	\label{algor: FRRSC}
\end{algorithm}

We firstly initialized $S$ the first $m$ complete tuples (Line 1-?). For each tuple $N_i$, if it is complete and it shares more than $\epsilon$ attributes with some tuple in $S$ (in Line ?-Line ?), it replaces some tuple in $S$ randomly (Line ?).
%
%

\textbf{Example 2:} We attempt to sample 7 popular items from the data set shown in Example 1. We first pick $t_{1-7}$ to $S$. Then, for $t_8$, we compare it with the samples in $S$. If we set $\epsilon$ as $2$, then we can find that ${\rm cmp}(S_i ,t)\geq \epsilon$ because:
\begin{displaymath}
	t_8[{\rm CC}]=t_3[{\rm CC}]; t_8[{\rm PN}]=t_3[{\rm PN}].
\end{displaymath}

\noindent Hence we generate a random number from $1$ to $8$. If we generate $2$, then $S_2$=$t_8$ rather than $t_2$.

After that, for $t_9$, we find that the item is incomplete and see it as a misleading tuple. Thus, we check $t_{10}$ without changing $q$. For $t_{10}$, we can find that it is also a misleading tuple since it is incomplete.

We check $t_{11}$, and find it is complete. However, when we compare it with samples pointed by $S$, we find that no sample can have more than two same attributes, which shows that it has the second feature of misleading tuples. Hence, it is also a misleading tuple. Since there is no more tuples for us to select, we obtain the samples: $t_1,t_8,t_3,t_4,t_5,$  $t_6,t_7$.

\textbf{Effectiveness Analysis} Theorem 1 shows the effectiveness of proposed algorithm.
\newtheorem{theorem}{Theorem}
\begin{theorem}
	The FRRSC can keep the probability of sampling for all popular tuples the same and avoid obtaining misleading tuples.
\end{theorem}
\begin{proof}
	We prove Theorem 1 by Mathematical Induction.
\end{proof}

\noindent\textbf{(1) Initialization}

When we just put the first $m$ tuples in $S$, the probability for each of the $m$ tuples appearing in $S$ is 1. When the $(m+1)^{\rm th}$ tuple is checked, if it is a popular tuple, we generate a random number in $[1,q+1]$, where $q$ is equal to $m$. If the number is in $[1,m]$, the $(m+1)^{\rm th}$ tuple can be added to $S$.

The probability of adding the $(m+1)^{th}$ tuple is $p_1=m/(m+1)$. For each sample $t$ in $S$, the probability that $s$ is replaced by $m+1$ is $p_2=p_1\times (1/m)=1/(m+1)$. Therefore, the probability that $t$ is not replaced by the $(m+1)^{th}$ tuple is $p_3=1-p_2=m/(m+1)$. Since $p_4=1$, $p_3\times p_4=m/(m+1)$.

\noindent\textbf{(2) Induction Assumption}

We suppose that when the $t$th popular and complete tuple is selected by probability $m/q$, the previous popular and complete $t-1$ tuples are sampled with the same probability $m/q$.

\noindent\textbf{(3) Induction}

The condition of the sample of the $(t+1)^{\rm th}$ tuple: It can be computed in two steps: (a) Before the $(t+1)^{\rm th}$ sample, the popular tuple is selected by the probability as $p_5=m/q$ according to the suppose. (b) If the $(t+1)^{\rm th}$ tuple does not replace the current samples in $T[1]$, we can know the probability is $p_6=q/(q+1)$. Because the probability of item being replaced is $p_7=1/(q+1)$, $p_6=1-p_7$. Since $p_6\times p_5=(m/q)\times (q/(q+1))=m/(q+1)$, this shows that suppose is right.

\noindent\textbf{(4) Conclusion}

According to initialization and induction, we know it is random to the popular tuples. And the misleading (unpopular/special) tuples cannot be extracted.
Time complexity Analysis To the first time of extraction, we can know if we have a data set with $N$ tuples and each samples in $T[1]$ except the front m tuples. When we find the same attributes between the two tuples are more than $b'$, we do not need to compare anymore. The average times of comparing maybe $(b'+r)/2$. Therefore, the time complexity is $O(n)$ (the times of comparing is $(b'+r)/2\times m\times N$).

\begin{algorithm}
	\caption{TRRSC}
	\hspace*{0.02in}{\bf Input:}
$ \boldsymbol{N} $ is big data set. $ \boldsymbol{m} $ samples in each group. $ \boldsymbol{b} $ and $ \boldsymbol{b'} $ set by us due to the data type and demand of user as standards of similarity. The samples $ \boldsymbol{T[1]} $ to $ \boldsymbol{T[i-1]} $ and each set is with $ m $ indexes from $ \boldsymbol{T[a,1]} $ to $ \boldsymbol{T[a,m]} (1\leq a\leq i-1)$.\\
	\hspace*{0.02in}{\bf Output:}
	The group of indexes from $ \boldsymbol{T[1]} $ to $ \boldsymbol{T[n]} $.
	\begin{algorithmic}[1]
		\State $ p=i\times m $;
		\State number $i\times m+1$ to $ N $ tuples from 1 to $N-i\times m$;
		\State $ t=1 $;$ q=1 $;
		\If{there is $ (t+1)^{\rm th} $ item exists in $ N $}
			\State $ t=t+1 $;
			\If{$ N[t] $ is complete}
				\For{$ i=1 $ to $ \min(q,m) $}
					\If{cmp$(T[1,i],t) \geq b$}
						\For{$ j=1 $ to $ i-1 $}
							\For{$ k=1 $ to $ m $}
								\If{(cmp$ (T[j,k],t) \geq b' $) \textbf{and} (cmp$ (T[j,k],t)\leq b $ (\textbf{for} all $ 1\leq  j \leq  i-1,$ $ 1\leq k\leq  m $))}
								\State $ q=q+1 $;$ k= $rand$ [1,q] $;
									\If{$ k\leq m $}\ $ T[i,k] $ point to $ N[t] $;
									\State go to 4;
									\EndIf
								\EndIf
							\EndFor
						\EndFor
						\Else
						\State go to 4;
					\EndIf
				\EndFor
			\EndIf
		\Else
			\If{q$ \geq $m}\ We start the next iteration;
			\Else
				\State $ n=i-1 $;
				\State output $ \boldsymbol{T[1]} $ to $ \boldsymbol{T[n]} $ as sampling result.
			\EndIf
		\EndIf
	\end{algorithmic}

	\hspace*{0.02in}{\bf Note:}
	cmp$ (T[1,i],t) $ shows the number of the same attributes when $ i=1 $; rand$ [1,q] $ is 1 when $ q=1 $ which guarantee the $1^{\rm th} $ tuple can be added as what we want.
	\label{algor: TRRSC}
\end{algorithm}

\subsubsection{TRRSC: Extraction of $ 2^{\rm nd}-n^{\rm th} $ Groups of Items}
\label{section: trrsc}

\textbf{Algorithm Overview} By calculating the number of the same attributes of the tuple with samples in $T[1],T[2],\cdots,$ $T[i-1]$, we ensure the samples are popular (there is a sample with no less than $b'$ same attributes) but different from the samples obtained in previous iterations (there is no sample with more than $b$ same attributes) and establish a new sample set for it.

\textbf{Algorithm Description} The pseudo code of the algorithm is shown in Algorithm \ref{algor: TRRSC}. Such function is invoked for $n-1$ times to generate $T[2]$ to $T[n]$.

When we choose the $i^{\rm th} (i\leq n)$ group of samples, we set $i\times m+1$ as the starting number firstly (Line 1-2) .The reason is that since we choose at least $m$ items each time, there are no popular items from $1$ to $i\times m$ in the $i^{\rm th}$ sampling.

We then obtain samples from $i\times m+1$ to $N$. We number the $i\times m+1$ to $N$ tuples as items from $1$ to $N-i\times m$. We also set the two variables $t$ and $q$ to show the number of tuples we are dealing with and the number of new kind of found popular tuples, respectively (Line 3). To generate $ T[i,1] $, we check tuples from the first one to validate whether they can meet our new standard.

New standard is to compare the tth tuple with the samples in $ T[1] $ to $ T[i-1] $ (Line 9). If a sample has more than $ b' $ same attributes with $ t $ and no sample has more than $ b $ same attributes with $ t $, we set $ k $ as $ 1 $ and add this tuple as the first one (Line 11). We check it to prevent samples from being too similar to make the CFDs strict. $ b $ is a high limit ensuring that the chosen tuple is not similar to those samples we choose and $ b' $ is a lower bound to ensure that chosen samples are not too special to make CFD useless.

Then, we continue to add new tuples. However, for each time, except the comparison with samples in $ T[1] $ to $ T[i-1] $, we compare each attribute in the $ t^{\rm th} $ tuple with each sample in $ T[i] $ (Line 7). If at least a sample in $ T[i] $ has more than $ b $ attributes same with attributes of the $ t^{\rm th} $ (Line 8), we compare it with samples in $ T[1] $ to $ T[i-1] $.

After that, we increase $ q $ by $ 1 $ and generate $ k $ in $ [1, q] $ (Line 12). We compare $ k $ with $ m $ to decide whether to replace sample in $ T[i] $ in the same way as FRRSC (Line 13). And if there is no sample in $ T[i] $ has at least $ b $ attributes same as $ t $ or some attribute of $ t $ is blank, or there is no sample in $ T[1] $ to $ T[i-1] $ following our new standard, we see it as a new tuple (Line 14).

Finally, when no new tuple is left, if we find the number of popular tuples similar to samples in $ T[i] $ represented by $ q>m $, we know that there are still new tuples not found. And we perform the $ (i+1)^{th} $ time of iteration to find new kind of tuple (Line 18). However, when we find $ q\leq m $, we can know that almost all kinds of popular tuples have been found. At this time, we can set n as $ i-1 $ since we only need $ i-1 $ times of sampling to find representative CFD sets. The $ i^{th} $ time of sampling is canceled (Line 20-21).

We use an example to demonstrate process of algorithm.

\textbf{Example 3:} If we have found a sampling set of $ T[1]={t_1,t_2,t_3,t_4} $ and want to find the second sampling set, we start from the $ (4\times 1+1)^{\rm th}=5^{th} $. We compare $ t_5 $ with samples in $ T[1] $.

If we set $ b' $ and $ b $ as $ 2 $ and $ 3 $ respectively, we can find that
\begin{displaymath}
	t_5[{\rm STR}]=t_3[{\rm STR}];t_5[{\rm ZIP}]=t_3[{\rm ZIP}].
\end{displaymath}

Since no samples in $ T[1] $ share 3 attributes with $ t_5 $, we add $ t_5 $ to $ T[2] $ as the first sample.

We can find that t6 has more than 3 attributes same as $ t_5 $. Then we compare $ t_6 $ with samples in $ T[1] $ and find that $ t_3 $ share 2 attributes with $ t_6 $ but no sample shares 3 attributes with it. Thus, we add $ t_6 $ to $ T[2] $. Then, we can find 3 attributes in $ t_7 $ the same as those in $ t_6 $. Meanwhile, $ t_1 $ in $ T[1] $ has two attributes same as $ t_7 $ and no item in $ T[1] $ has 3 attributes same as $ t_7 $. Then, we add $ t_7 $ to $ T[2] $. Since we find that no item in $ T[2] $ has 3 attributes same as $ t_8 $, we give it up and turn to $ t_9 $. Then, we find $ t_9 $ and $ t_{10} $ are incomplete and $ t_{11} $ is special.

Therefore, we can get $ T[2]={t_5,t_6,t_7} $ which is too small. So we quit $ T[2] $ and let $ n=1 $ with $ T[1] $ as sampling result.

\textbf{Effectiveness Analysis} Theorem 2 shows the effectiveness of proposed algorithm.
\begin{theorem}
	For popular items similar to the sampling set $ T[i] $, we ensure their probability be sampled the same in the ith sampling and avoid sampling misleading items in TRRSC.
\end{theorem}
\begin{proof}
	We can utilize proof in FRRSC. Each iteration is similar to the $ 1^{\rm st} $ one. In the process of $ 2^{\rm nd}¨Cn^{\rm th} $ times of sampling, we only add a new comparison with items we have sampled. The rule of random number generation keeps the same. We can see $ 2^{\rm nd}¨Cn^{\rm th} $ times of sampling composed of $ n-1 $ times of $ 1^{\rm st} $ extracting. Since each time is random sampling, we can know it is random for the whole process. And the probability is zero to special and incomplete items.
\end{proof}

\textbf{Time complexity Analysis} To the process of $ 2^{\rm nd}-n^{\rm th} $  times of sampling, we can know that: For the ith sample, we need to compare each item with items in $ T[1],T[2],\cdots,T[i-1] $. Therefore, we need to compare for $ (i-1)\times m $ times. Total times are $ (i-1)\times m\times N\times r $ for the ith extraction. Therefore, total times of comparing is:. To the process of I/O, the data is scanned for one time making us see time complexity as time s of comparison. Thus, the time complexity is $ O(n) $.

\subsection{One-Pass Sampling Algorithm}
\label{section: one-pass}

\textbf{Algorithm Overview} For a big data set, we should compass all iterations in one scan to save time. Initially, we make the $ m $ indexes in $ T[1] $ point to the first $ m $ tuples and establish an array $ q $. Each element $ q[i] $ is the number of the tuples similar to $ T[i] $.

Then, we compare each new scanned tuple with samples in our sampling sets. If a sample in $ T[i] $ has more than $ b $ same attributes with it, we add $ q[i] $ by $ 1 $ and add it into $ T[i] $ if $ q[i]<m $. When $ q[i]\geq m $, we generate a random number $ k $ in $ (1,q[i]) $. If $ k $ is no larger than $ m $, we add it as the $ k^{\rm th} $ sample in $ T[i] $. Otherwise, we abandon it.

If no sample has more than $ b $ same attributes with the tuple, we check whether a sample has no less than $ b' $ same attributes with it since it may be special. If there is such a sample, we know that it is not special and put it into $ T[i+1] $. Otherwise, it will be abandoned.

When sampling from real big data, we observed that the possibility of popular tuple being sampled is too small. If we firstly generate a random number $ k $ and compare attributes only when $ k $ is no larger than $ n\times m $, we will save many comparing times. As the cost, we will lose some tuples when counting items similar to $ T[a] $. It is because that even though a new tuple is similar to $ T[a] $, we do not know whether it is similar or not without comparing it with tuples in sample sets when $ k>m $. This will make us delete the $ T[a] $ wrongly since the amount of its similar tuples is smaller than $ m $. For big data, $ T[a] $ always has more than $ m $ similar tuples. Therefore, after all reservoirs are full $ (\min(q[a])\geq m\ (0<a\leq n)) $ ,we can generate a random number before comparing new tuple with other samples.

\textbf{Algorithm Description} Pseudo code is shown in Algorithm 3. We first set $ m $ pointers in table $ T[1] $ pointing to items from $ 1 $ to $ m $, initialize a variable $ t $ and an array $ q[n] $ (Line 2-4). $ q[i] $ is number of tuples similar to those samples in $ T[i]$. $t$ is increased by $ 1 $ and when there is a reservoir not full ($\min(q[a])<m\ (0<a\leq n)$)  we compare each attribute in $ N[t] $ with samples in $ T[1]\cdots T[i] $ (Line 9-10).

If there is at least one sample in $ T[a] $ having more than $ b $ attributes with the same amount as the attributes of $ N[t] $ (Line 11),we increase $ q[a] $ by 1 and generate a random integer $ k $ in $ [1,q[a]] $ when $ q[a]\geq m $ (Line 12-13). When $ k\leq m $, we replace sample $ T[a,k] $ with $ N[t] $ (Line 14-15). When $ k>m $, we find a new tuple. When $ q[a]<m $, we add $ t $ as $ T[a,q[a]+1] $ directly (Line 17).

When we compare $ N[t] $ with samples in $ T[1] $, $ T[2] \cdots T[i]$, we also check whether there is an item having more than $ b' $ attributes same with $ N[t] $ and set label as 1 to show there is such an item. If there is no item having more than b same attributes with $ N[t] $, we check whether the label is 1. If label is 1 showing that some sample having more than $ b' $ same attributes with $ t $, we build a new group $ T[i+1] $ and set it as $ T[i+1,1] $ (Line 21-22).

When all reservoirs are full ($\min(q[a])\geq m\ (0<a\leq n) $), we generate a random integer $ k $ in $ [1,t] $, and compare each attribute in $ N[t] $ with samples in $ T[1],T[2],\cdots,T[i] $ (Line 29-30) only when $ k\leq n\times m $ (Line 27). We use $ n\times m $ rather than $ m $ as the high limit because there are $ n $ sample sets. Then, if there is at least one sample in $ T[a] $ having more than $ b $ attributes with the same amount as the attributes of $ N[t] $, we increase $ q[a] $ by 1 and replace the sample $ T[a, k\%m] $ with $ N[t] $ (Line 33).  When comparing, we also let label equal to 1 to show there is such an item having more than b' attributes same as $ N[t] $. We build a new group $ T[i+1] $ and set it as the $ T[i+1,1] $ (Line 37).Therefore, we synthesize the two phases in FRRSC and TRRSC in once scan. Finally, the results are $ T[1],T[2],\cdots,T[n] $ (Line 41).

\textbf{Example 4:} We show an example to sample from small data set, in which condition, $ \min_{0<a\leq n} (q[a]) $ is always smaller than $ m $. As to big data, it is hard to show due to the limit of space. To sample from big data, when $ \min_{0<a\leq n}(q[a])<m $, the sampling process keeps the same. When $ \min_{0<a\leq n}(q[a])$  $\geq m $, the sampling process is very similar. Hence we do not show how to sample from big data.

We compare t6 with $ T[1] $ and find that no sample in $ T[1] $ has more than 3 attributes same with it. However, when comparing it with $ T[2] $, we find that it has 5 attributes same with $ T[2,1] $ which is $ t_5 $ actually. Then since $ q[2]=1<3 $, we insert the $ t_6 $ directly to $ T[2,2] $.

When it comes to $ t_7 $, we compare it with $ T[1] $ and the result is the same as $ t_5 $ and $ t_6 $. But when we compare it with $ T[2] $, we find it has 3 attributes same as $ t_6 $. Meanwhile, since $ q[2]=2<3 $, we add the $ t_7 $ to $ T[2,3] $ directly.

As to $ t_8 $, we find that no sample in $ T[1] $ and $ T[2] $ has more than 3 attributes same with its. However, it has 2 attributes same with $ t3 $ tuple. Therefore, we add it to $ T[3,1] $.

To $ t_9 $ and $ t_{S} $, we can find both of the two items are incomplete, which makes we abandon them directly. After that, we find that no item in $ T[1] $, $ T[2] $ and $ T[3] $ has more than 2 attributes same with $ t_{11} $'s attributes. Finally, we check $ T[1] $, $ T[2] $ and $ T[3] $, and we find that
S\begin{displaymath}
	q[1]=4\geq 3;q[2]=3\geq 3;q[3]=1<3.
\end{displaymath}

\noindent Hence we abandon $ T[3] $ and leave $ T[1]=\{t_2,t_3,t_4\},T[2]=\{t_5,t_6,t_7\} $ as sampling results. $ n=2 $ is the number of groups.

\textbf{Effectiveness Analysis} Theorem 3 shows the effectiveness of the proposed algorithm.

\begin{theorem}
	For the popular complete tuples in a big data set which are all similar to the same $ T[i] $ sampling set, the probability of extraction keeps the same in BRRSC. And the misleading tuples cannot be extracted in BRRSC.
\end{theorem}
\begin{proof}
	We firstly put the first $ m $ tuples in the $ T[1] $. Then if $ \min(q[a])<m\ (0<a\leq n) $, for each time we add the tuple to $ T[i] $ or establish a new sample set $ T[n+1] $. When $ q[i]\geq m $, we generate a random number in $ [1,q[i]] $. Therefore, for $ T[i]\ (1\leq i\leq n) $, the condition of sampling similar samples is similar to the first extraction for $ T[1] $. Only the condition of generating a random and add $ q $ with 1 is different, and this does not influence the calculation of probability and the result of equal probability. Meanwhile, if $ \min(q[a])\geq m\ (0<a\leq n) $, we just change the order of sampling processes which will not influence the probability.
\end{proof}

\textbf{Time complexity Analysis} Different from the $ 2^{\rm nd}-n^{\rm th} $ extraction, we do not have to compare with all the sampled items to ensure the item is new. We can add it to its similar $ T[i] $ directly. When $ \min(q[a])<m $, since the average times comparing with sampling items is $ (n/2) $. Hence the complexity of time $ f_1(|N|) $ is $f_1(|N|) =r\times m\times (n/2) \times |Nf|=O(c) $.

$ N_f $ is a small part of $ N $, which can make each sample set $ T[i] $ have more than $ m $ items. When $ \min(q[a])\geq m $, we firstly generate $ k $ in $ [1,t] $ before we compare item's attributes. Probably we can compare attributes is $ p_1=(m\times n)/t $ Therefore, for $ N_b $ which shows a large part of $ N $ except $ N_f $,  the time complexity $ f_2(|N|) $ is

\begin{align*}
	f_2(N)=&(\frac{1}{\left|N_f\right|}+\frac{1}{\left|N_f\right|+1}+\frac{1}{\left|N_f\right|+2}+\frac{1}{\left|N_f\right|+3}+\cdots\\
	&+\frac{1}{\left|N\right|})\times m\times n\times r\times(n/2)\\
	=&(\ln(\left|N\right|)-a)\times m\times n\times r\times(n/2)\\
	=&O(\ln(\left|N\right|)).
\end{align*}

\noindent The complexity of time $ f(\left|N\right|) $ for the whole process is

\begin{align*}
	f(\left|N\right|)&=f_1(\left|N\right|)+f_2(\left|N\right|)\\
	&=(\ln(\left|N\right|)-a+\left|N_f\right|)\times m\times n\times r\times(n/2)\\
	&=O(c)+O(\ln(\left|N\right|))\\
	&=O(\ln(\left|N\right|).
\end{align*}

This shows that the complexity of sampling is $ O(\ln(|N|)) $ which is sub-linear to the data set.

\begin{algorithm}[!h]
	\caption{BRRSC}
	\hspace*{0.02in}{\bf Input:}
	Data set $ \boldsymbol{N} $. $ \boldsymbol{m} $ samples in each group. $ \boldsymbol{b} $ and $ \boldsymbol{b'} $ set by us due to data type and demand of user as standards of similarity.\\
	\hspace*{0.02in}{\bf Output:}
	The groups of indexes $ \boldsymbol{T[1]} $ to $ \boldsymbol{T[n]} $.
	\begin{algorithmic}[1]
		\For{$ w=1 $ to $ m $}
			\State $ T[1,w] $ point to $ N[w];$
		\EndFor
		\State $ t=m; $ label$ =0 $;
		\State $ q[1]=m $;
		\If{there is $ N[t+1] $}
			\State $ t=t+1; $
			\If{$ \min_{0<a\leq n}(q[a])<m $}
				\If{$ N[t] $ is complete}
					\For{$ w=1 $ to $ i $}
						\For{$ j=1 $ to $ \min(m,q[w]) $}
							\If{cmp$ (T[w,j],t)\geq b $}
								\If{$ q[b]>m $}
									\State $ q[b]=q[b]+1;k= $rand$ [1,q[b]]; $
									\If{$ k\leq m $}
										\State $ T[b,k] $ point to $ N[t]; $go to 5;
									\EndIf
								\Else
									\State $ q[b]=q[b]+1;T[b,q[b]] $ point to $ N[t] $;
									\State go to 5;
								\EndIf
							\ElsIf{cmp$ (T[b,j],t)\geq b' $}
								\State label=1;
							\EndIf
						\EndFor
					\EndFor
					\If{label==1}
						\State $ T[i,1] $ point to $ N[t]; $go to 5;
					\EndIf
				\Else
					\State go to 5;
				\EndIf
			\ElsIf{$ \min_{0<a\leq n}(q[a])\geq m $}
				\State $ k= $rand$ [1,t]; $
				\If{$ k\leq m\times n $}
					\If{$ N[t] $ is complete}
						\For{$ w=1 $ to $ i $}
							\For{$ j=1 $ to $ \min(m,q[w]) $}
								\If{cmp$ (T[w,j],t)\geq b $}
									\State $ q[b]=q[b]+1; $
									\State $ T[b,k\%m] $ point to $ N[t] $; go to 5;
								\ElsIf{cmp$ (T[b,j],t)\geq b'$}
									\State label=1;
								\EndIf
							\EndFor
						\EndFor
						\If{label==1}
							\State $ T[i,1] $ point to $ N[t]; $go to 5;
						\EndIf
					\Else
						\State go to 5;
					\EndIf
				\EndIf
			\Else
				\State output $ T[1],T[2],\cdots,T[n]; $
			\EndIf
		\EndIf
	\end{algorithmic}
	\label{algor: BRRSC}
\end{algorithm}

\section{BDC: CFD Discovery For Big Data}
\label{section: bdc}

After sampling, we need to find rules on n small data sets. For the discovery, we still have following problems to solve.

\begin{enumerate}
	\item[(1)] Although we use the RRSC, there may also be some special or dirty samples. The CFD discovery algorithm should be fault-tolerant.
	
	\item[(2)] Since there are variable kinds of big data, we need to make our method fit different conditions. Meanwhile, we need to ensure the CFD set is complete. Therefore, our method should discover both constant and variable CFDs and tolerate faults. Such algorithm is in Section \ref{section: dfcfd}.
	
	\item[(3)] Due to errors in the training set, it is possible to find conflicts in CFDs produced by an algorithm. To resolve the conflicts, we establish a graph-based method to find correct CFDs by finding disconnected subset with largest weights in Section \ref{section: deal with confict}.
\end{enumerate}

\subsection{DFCFD: Dynamical Fault-tolerant CFD Discovery Algorithm}
\label{section: dfcfd}

\textbf{Algorithm Overview} DFCFD is designed to find CFDs from the results of sampling. We improve three CFD discovery algorithms CTANE, FastCFD, CFDMiner [1] to BCTANE, BFCFD and BCFDM by accepting some CFDs with limited confidence to tolerant fault. We find differe-nt algorithms have their preference to variable big data, so we choose different groups of algorithms. During synthesis, we utilize the same process of different methods.

The whole work of DFCFD algorithm is shown in Figure 2. It can be found after getting samples, we preprocess samples. And we have two choices of algorithm combination which will be introduced in the following.

\textbf{Algorithm Description} We then introduce the detailed information of three improved algorithms and the synthesis of them.

\textbf{BCTANE} To improve CTANE, we use a threshold $ e $ to decide whether we can accept a CFD. For each CFD, we set a variable $ u'=|T| $ ($ T\subseteq r $ and CFD is absolutely right for items in $ T $). $ |T| $ denotes the number of the samples in $ T $ which is a set of samples. $ r $ is a sample set where we find CFDs. Then we get a new variable $ u=u'/|T'| $ ($ T'\subseteq r $ it conforms to left side (premise) of CFD). We improve CTANE by adding following two steps.

\begin{enumerate}
	\item[a.] When we cut a limb, we change the rule to that if $ u_{\rm CFD}\leq e $, then we cut the limb.
	
	\item[b.] When we calculate the supporters for a CFD, we think that items with same LHS can support CFD when RHS is empty or wrong (means similarity$ >e $).
\end{enumerate}

\begin{figure}[h]
	\centering
	\includegraphics[width=6cm]{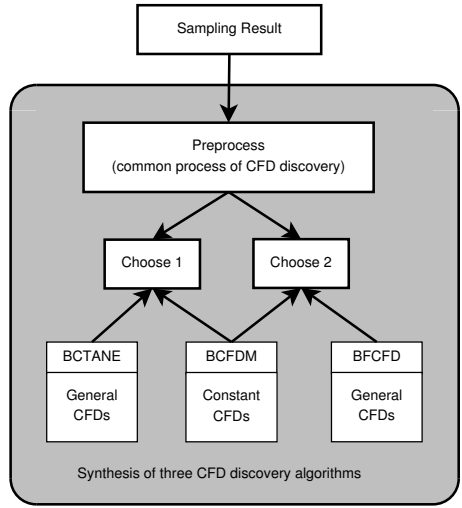}
	\caption{The overview of DFCFD algorithm}
	\label{fig: DFCFD}
\end{figure}

\textbf{BFCFD} To develop FastCFD, we change its procedure FindMin to adapt to data sets with special or dirty ones. When FindMin determines whether a constant ta makes constant CFD$ (X\rightarrow A,(t_P\parallel t_a)) $ valid, we check whether there is no $ X'\subseteq X $ in size $ |X|-1 $ making CFD$ (X'\rightarrow A,(t_P[X']\parallel t_a)) $ valid in FastCFD. However, when it comes to big data, many samples can contain errors or incomplete one. Hence we make the BFCFD allow some different items to make CFD$ (X'\rightarrow A,(t_P [X']\parallel t_a)) $ valid, when following constraint is satisfied.

\begin{align*}
	&u'=|T|(T\subseteq r; CFD(X'->A,(t_P[X']||t_a) {\rm \ is\ right}\\
	&\quad {\rm for\ items\ in}\ T)\\
	&u=u'/|T'|(T'\subseteq r {\rm \ and\ it\ conforms\ to\ } t_P[X'] )
\end{align*}

For the constant CFDs, when $ u > e $, we say that CFD $ (X'\rightarrow A,(t_P[X']\parallel t_a)) $ is valid and acceptable.

Then, in FindMin, to find variable CFDs from big data, we use a threshold of error $ e $ to tolerant the wrong samples, we revise the constraints as follows.

\begin{enumerate}
	\item[a'.] If number of $ X'\subseteq X $ in size $ |X|-1 $ making  $Y\cup(X\setminus X')$ cover $D_A^m(r_{t_P[X']})$  is less than $ e\times I $ .
	
	\item[b'.] If number of $ Y'\subseteq Y $ of size $ |Y|-1 $ making $ Y' $ covering  $D_A^m(r_{t_P[X]})$ is less than $ e\times i' $.
\end{enumerate}

If a' and b' are both satisfied, the variable CFD is accepted.

\textbf{BCFDM} We change the CFDMiner in the way similar to above two improvements. In CFDMiner's third step, we check the free item set $ (Y,s_p) $ in list $ L $ with following constraints (the number of attributes in $ Y $ is shown by $ i $).

\begin{enumerate}
	\item[a.] For each subset $ Y'\not\subset Y $ such that $ (Y',s_p[Y'])\not\subset L $, we replace RHS$ (Y,s_p) $ with RHS$ (Y',s_p[Y']) $. But the RHS$ (Y',s_p[Y']) $ cannot lead to a left-reduced constant CFD.
\end{enumerate}

For big data, we can ignore these wrong tuples and constraint is modified as follows.

\begin{enumerate}
	\item[a'.] If number of the subsets of $ Y'\not\subset Y $ and RHS$ (Y',s_p[Y']) $ leading to that a left-reduced constant CFD is less than $ e\times i $. Or when comparing the items similar to wrong item, if the similarity of similar items and wrong one is larger than $ e $, we gather similar with the wrong one to find there is no left-reduced constant CFD. If condition $ a' $ is met, we can accept $ (Y,s_p) $.
	
\end{enumerate}

\textbf{Integration of three algorithms} In order to synthesize these three algorithms, we should merge the same or similar processes of these three methods to accelerate the whole process by preprocessing. According to [1], these three original algorithms all need to know the supporters of different attribute sets, which is same to our improved algorithms. Therefore, we firstly generate the number of supporters for different attributes and put them in a hash table. Then, using the hash table, we can reduce repeat calculation in the process of finding CFDs by three algorithms.

To choose the algorithms, we need to consider different preference of them. Since we have not changed a lot about three algorithms, the function of the improved algorithms is similar to that of original ones. Then according to [1], we can find CTANE cannot run to completion when arity is above 17 and it can be sensitive to support threshold and outperforms FastCFD when the data set is large with small arity. However, FastCFD can outperform CTANE when arity is larger than 17 and can do well for small data set with few attributes. What is more, CFDMiner can always outperforms the other two by three orders of magnitude making us ignore its efficiency. Therefore, we choose BCTANE and BCFDM when arity is smaller than 17 and items are more than a million. When arity is larger than 17, we utilize BFCFD and BCFDM together.

\subsection{Deal with conflicts between CFDs}
\label{section: deal with confict}

With dirty data in the training set, the discovered CFDs may contain conflicts. Since we premise that the large part of data set is clean, we attempt to find a maximum compatible rule subset. Thus, we model the CFD set as a weighted undirected graph including CFDs as nodes. We add a line between two nodes when there is conflict between two CFDs. The weight of each node is the number of supporters for each node. Then the problem of finding maximum compatible rule subset is converted to finding maximal weight independent set of nodes from graph. To solve it, we develop linking rules and MWID algorithm. In this section, we first introduce how to get the weight of each node (Section \ref{section: cal weight}), then we represent the conflicts between CFDs by linking rules (Section \ref{section: disc confict}), and finally we use MWID algorithm to find a maximal weight independent set (Section \ref{section: mwid}).

\subsubsection{Calculating the weight of each node}
\label{section: cal weight}

We use the number of supporters of a CFD as weight of each node in WCFD. The WCFD is a weighted undirected graph for CFDs. For constant CFDS, such number could be computed by SQL, while it is harder for variable CFDs.

Thus, we propose a new method to calculate variable CFDs' supporters. We firstly build a rank for the number $ (r_1,r_2),(r_2,r_3),(r_3,r_4),\cdots $ for the samples with $ n $ samples in it. We should notice that the ranker has larger distance in back. And when it comes to the half of $ n $, we think the supporters as large enough to ignore the difference between them. Thus, we can set the last rank as $ (n/2,n) $.

With the rank, we can set the threshold k instead of e in finding CFDs by FastCFD or CTANE as the $ r_1,r_2,r_3,\cdots $. If a CFD exists in CFD set for $ k=r_i $ and does not exist in the CFD set for $ k=r_{i+1} $, we can set the amount of supporters for the CFD as $ {\rm int}[(r_i+r_{i+1})/2] $. But if it reaches the final rank, we use 80\% of n as its supporters.

\subsubsection{Discovery of the conflict between two CFDs}
\label{section: disc confict}

When decide whether there is conflict between two CFDs, we design a deciding rule- Linking Rule. By such rule, we can decide whether to set a line between two CFD nodes to show conflict between them. We discuss linking rules in two cases with two CFDs and multiple CFDs.

\textbf{For two CFDs:} $ C_1 $: $ (X_1\rightarrow A_1,(t_P [X_1]\parallel t_1)) $; $ C_2 $: $ (X_2\rightarrow A_2,(t_P [X_2]||t_2)) $. We firstly decide whether there is conflict between $ C_1 $ and $ C_2 $. We can divide the problems into three situations according to the relationship between $ X_1 $ and $ X_2 $. Without generality, we suppose $ |X_1|\leq |X_2| $.

\textbf{T1 $ \boldsymbol{X_1\subset X_2} $.} Only if $ A_1 $ is the same as $ A_2 $, can there be conflict.

\textbf{T1-1.} If $ C_1 $ and $ C_2 $ are both constant CFD, then only when $ t_p[X_1C_1]=t_p[X_1C_2] $ but the $ t_p[A_1C_1]\not= t_p[A_2C_2] $, is there a conflict between them. Here, $ t_p[X_1C_1] $ and $ t_p[X_1C_2] $ means the range of the attribute set $ X_1 $ in $ C_1 $ and $ C_2 $ which is same for other attributes. e.g., $ C_1 $: $ (F,G\rightarrow A,(1,2\parallel 1)) $ and $ C_2 $: $ (F,G,H\rightarrow A,(1, 2,3\parallel 3)) $

\textbf{T1-2.} If $ C_1 $ and $ C_2 $ are both variable CFD, then when ``\underline{\quad}'' is for different attribute, there can be conflict. There must be at least one attribute $ r_i $ in $ X_1 $ that is a variable attribute with ``\underline{\quad}'' for its range and a constant data for $ r_i $ in $ X_2 $ to make a conflict. e.g.,$ C_1 $: $(F,G\rightarrow A,(\underline{\quad}, 2\parallel\underline{\quad}))$ and $ C_2 $: $ (F,G,H\rightarrow A,(1, 2,\underline{\quad}\parallel\underline{\quad}) ) $. We know that for $ C_1 $, when $ F $ is 1, $ A $ is a constant. However, from $ C_2 $, we know when $ F=1 $ and $ H $ is changed, $ A $ is changed with $ H $.

\textbf{T1-3.} If $ C_1 $ is a variable and $ C_2 $ is a constant, there cannot be conflict between two CFDs. Because when $ X1\subset X $ and $ C_1 $ is variable, the $ C_2 $ can be a kind of situation of it.

\textbf{T1-4.} If $ C_1 $ is a constant and $ C_2 $ is a variable, when $ t_p[X_1C_1] = t_p[X_1C_2] $, but in $ X_2 $ there is a variable attribute not in $ X_1 $. This results in that when $ A_1=A_2 $, $ A_2 $ is more general than $ A_1 $. Then there must be a conflict. e.g. , Rules $ C_1 $: $ (F,G\rightarrow A,(1, 2\parallel 2)) $ and $ C_2 $: $ (F,G,H\rightarrow A,(1, 2,\underline{\quad}\parallel\underline{\quad})) $ We know that when $ F=1 $, $ G=2 $, $ A $ in $ C_1 $ should be a constant. However, it is a variable with different $ H $. Then $ C_1 $ and $ C_2 $ conflict.

\textbf{T2 $ \boldsymbol{X_1=X_2} $.} Only if $ A_1 $ is the same as $ A_2 $, can there be conflict.

\textbf{T2-1.} If $ C_1 $ and $ C_2 $ are both constant CFD, then only $ t_p[X_1C_1]=t_p[X_2C_2] $ but $ t_p[A_1C_1]\not=t_p[A_2C_2] $ can imply a conflict.e.g., $ C_1 $: $ (F,G\rightarrow A,(1, 2\parallel 1)) $ and $ C_2 $: $ (F,G\rightarrow A,(1,2\parallel 3)) $.

\textbf{T2-2.} If $ C_1 $ and $ C_2 $ are both variable CFD, then when ``\underline{\quad}'' is for different attribute, there can be conflict. e.g., $ C_1 $: $ (F,G\rightarrow A,(\underline{\quad}, 2\parallel\underline{\quad})) $ and $ C_2 $: $ (F,G \rightarrow A,(1,\underline{\quad}\parallel\underline{\quad})) $. For $ C_1 $, when $ F $ is 1, $ A $ is a constant. However, from $ C_2 $, when $ F=1 $, and $ G $ is different, the $ A $ can change.

\textbf{T2-3.} If $ C_1 $ is variable and $ C_2 $ is constant, it cannot generate conflict for $ C_2 $ can be treated as a special situation for $ C_1 $.

\textbf{T3 $ \boldsymbol{X_1\subset X_2} $.} In this case, there can be conflict only when $ A_1 $ is the same to $ A_2 $. If $ X_1\cap X_2=\emptyset $ , it is unnecessary to compare these CFDs. So $ X_1\cap X_2=\emptyset $ should be satisfied to find a conflict. We suppose $ X_1\cap X_2=E $ , where $ E $ is attribute set shared by $ X_1 $ and $ X_2 $.

\textbf{T3-1.} If the $ C_1 $ and $ C_2 $ are both constant CFD, there cannot be conflict between the two CFDs. Since they cannot include the situation of the other, there cannot be a conflict.

\textbf{T3-2.} If $ C_1 $ and $ C_2 $ are both variable CFD, then when ``\underline{\quad}'' is the range for all the attributes in one CFD and in another CFD there are attributes in E whose range is fixed with attributes not in E whose range is fixed. e.g.,$ C_1 $: $ (F,G,H\rightarrow A,(\underline{\quad}, \underline{\quad},\underline{\quad}\parallel\underline{\quad})) $ and $ C_2 $: $ (F,L,Q\rightarrow A,(1,2, \underline{\quad}\parallel\underline{\quad})) $. For $ C_1 $ when $ F=1,G=2,H=8 $, $ A $ is a constant. From $ C_2 $, we can know that when $ F=1,G=2,H=8 $ but $ Q\not=H $. Thus, $ A $ in $ C_2 $ is different.

\textbf{T3-3.} If one CFD is a variable and another CFD is a constant, there cannot be conflict between them, since constant CFD can be seen as a special case for the other CFD when $ X_1\not\subset X_2 $.

\begin{figure}[h]
	\centering
	\includegraphics[width=4cm]{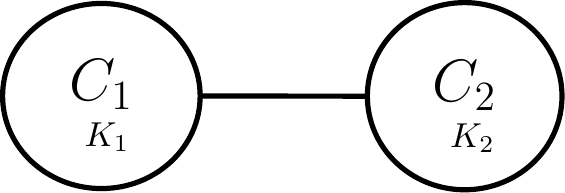}
	\caption{Build a line between $ C_1 $ and $ C_2 $  when there is conflict.}
	\label{fig: line}
\end{figure}

\begin{figure}[h]
	\centering
	\includegraphics[width=8cm]{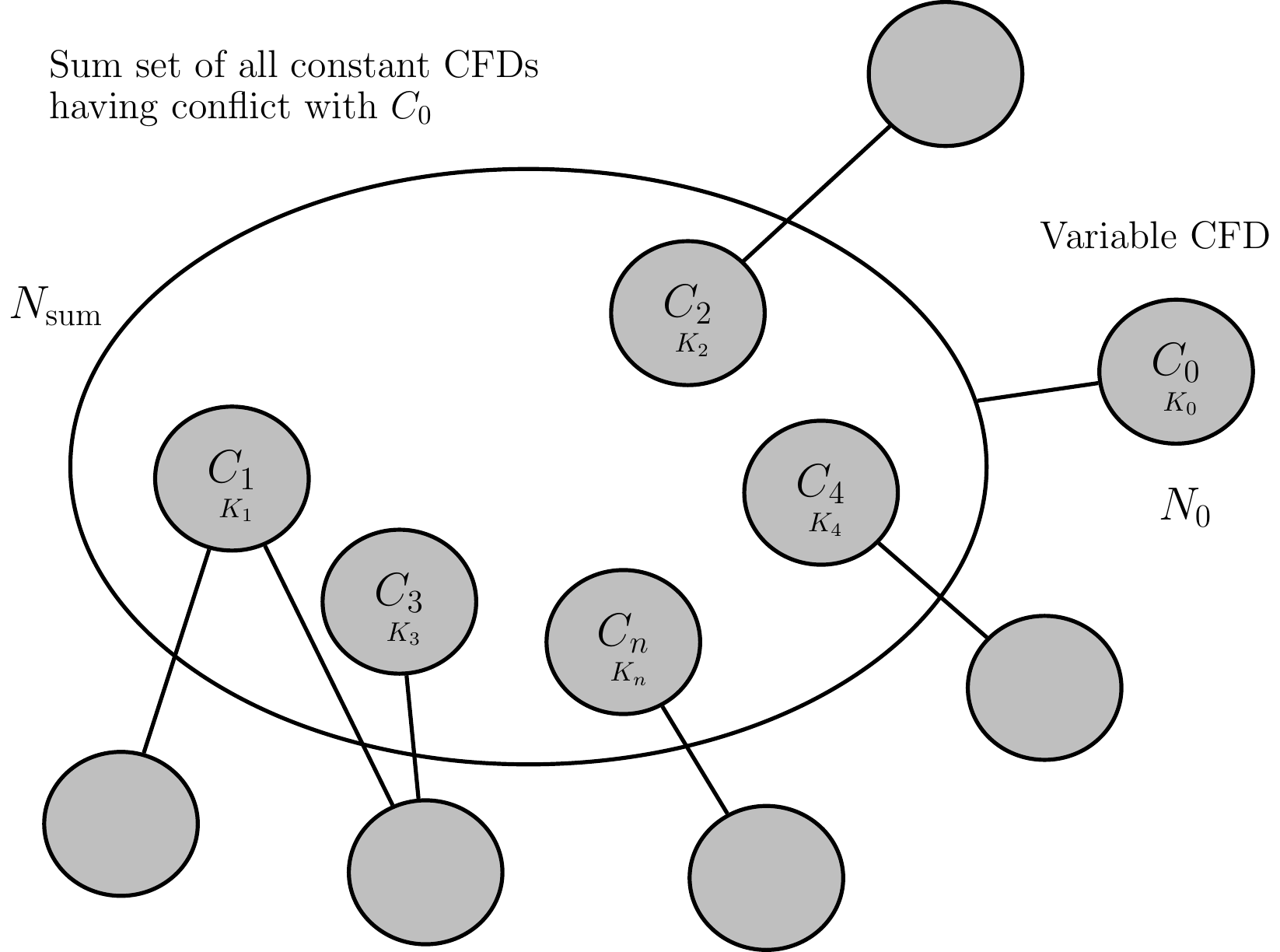}
	\caption{Put constant CFDs together and leave the variable CFD alone.}
	\label{fig: lines}
\end{figure}

\textbf{For more than two CFDs} When we find conflict among more than two CFDs, we can integrate the conditions of generating conflict into a rule M1. The only condition generating conflict is that for a variable CFD, there are no less than two constant CFDs showing that it is wrong. We suppose there are three CFDs-sets, which contains a variable CFD and two constant CFDs.
\begin{align*}
	&C_1: (X_1\rightarrow A_1,(t_P[X_1]\parallel t_1)); \\
	&C_2: (X_2\rightarrow A_2,(t_P[X_2]\parallel t_2)); \\
	&C_3: (X_3\rightarrow A_3,(t_P[X_3]\parallel t_3)).
\end{align*}

\textbf{M1.} If there is conflict among them, $ A_1$, $A_2 $ and $ A_3 $ must be the same attribute. There is at least one attribute shared by $ X_1$, $ X_2 $ and $ X_3 $. We denote such attribute set by $ U $. Meanwhile, in one CFD, the range of $ U $ is ``\underline{\quad}'' which means variable and $ A $ in this CFD is also a variable. But in other CFDs, $ U $ and $ A $ are both constants. Then we suppose that $ C_1 $ is a variable while $ C_2 $ and $ C_3 $ are both constant. We find that we can synthesize different conditions: $\{X_1\not\subset X_2,X_1\subset X_3 X_1=X_2 \} $ and all other conditions in one rule: Let $ E= X_1\cap X_2\cap X_3 $, then if one attribute in $ E $ is``\underline{\quad}'' for $ C_1 $ and it is the same constant data for $ C_2 $ and $ C_3 $. To the other attributes in $ E $, the range of them is the same for three CFDs. Then if $ A $ in $ C_2 $ is different from $ C_3 $, there is a conflict.

\textbf{According to Rule M1:} In the case that $ X_1\not\subset X_2\not\subset X_3 $, consider three rules $ C_1 $: $ (F,G,H\rightarrow A,(\underline{\quad}, 1,2\parallel\underline{\quad})) $, $ C_2 $: $ (F,G,Q \rightarrow A,(1,1,4\parallel 1)) $ and $ C_3 $: $ (F,G,W \rightarrow A,(1,1,$  $4\parallel 2)) $. We can discover from $ C_1 $ that when $ G=1 $ and $ H=2 $, $ F $ can decide $ A $. However, in $ C_2 $ and $ C_3 $, we discover that when $ G=1,H=2 $, $ F $ cannot decide $ A $. In the case that $ X_1\subset X_2=X_3 $, consider rules $ C_1 $:$ (F,G\rightarrow A,(\underline{\quad},2\parallel\underline{\quad}))$, $ C_2 $:$ (F,G,Q\rightarrow A,(1,1,4\parallel 1)) $ and $ C_3 $:$ (F,G,W\rightarrow A,(1,1,4\parallel 2)) $; We can discover from above discussions that $ F $ cannot decide $ A $ when $ G=2 $ by itself. There is a conflict between them.

For all the different relationship among $ X_1 $, $ X_2 $ and $ X_3 $, we can see that the rule M1 can work for all the conditions. Because if we want to see the conflict among more than two CFDs, we can only get the conflict when one CFD is variable and the others are constant. However, the constant CFDs of the others cannot show the variable CFD. Therefore, no matter what kind of relationship among $ X_1 $, $ X_2 $ and $ X_3 $, we can always check conflict by M1.

For the conflict between two CFD $ C_1 $ and $ C_2 $, we can just build a line between them like Figure 3. However, for more than two CFD nodes, we need to put constant CFDs together as a new node and leave variable CFD alone. The weight a combined node Nsum in Figure 4 is $ K_{\rm sum}=K_1+K_2+K_3+K_4+\cdots+K_n$. Other CFDs having conflict with $ C1,C2,\cdots,Cn $ also have conflict with the Nsum in Figure 4.

\subsubsection{MWID: Maximal weight independent discovery}
\label{section: mwid}

Since the premise for our method of finding CFDs from big dirty data is that the large part of data set is clean, we attempt to find a maximum compatible rule subset. Then with the maximum subset, we can cover the largest number of tuples in big data set. Since the maximal independent discovery problem, an NP-Hard problem \cite{Garey1986Computers}, is a special case with this problem with the weight of each vertex as 1, the maximal weight independent set (MWIS) discovery problem is also an NP-hard problem.

To find the MWIS from an undirected graph, we design an algorithm, MWID by improving algorithm FastMIS in \cite{Bollob1998Modern}. FastMIS introduces a randomized algorithm to find maximal independent set (MIS). It compute a MIS in a distributed way. However, the MIS computed by it contains the largest number of nodes and does not consider the weight. Therefore, we modify some steps in the FastMIS to generate the MWIS. In FastMIS, there are three steps to get MIS. The first two steps are as following.

\begin{enumerate}
	\item[a.] Each node $ v $ chooses a random value $ r(v)\in [0,1] $ and sends it to its neighbors.
	
	\item[b.] If $ r(v)< r(w) $ for all neighbors $ w\in N(v) $, node $ v $ enters MIS and informs its neighbors.
\end{enumerate}

The two steps make sure that if a node $ v $ joins the MIS, then $ v $'s neighbors do not join MIS at the same time. By this method, the node with the globally smallest value will always join the MIS to find maximal independent set, which has been proved in \cite{Bollob1998Modern}. When considering the weight, we need to ensure the nodes with larger weight having more possibility to join the MIS. So we cannot let the range to select a random value keeps the same. The modified algorithm is as following.

\begin{enumerate}
\item[a'.] Compare the weight $ w_v $ of each node $ v $ with each weight wn of its neighbors. If $ w_v>w_n $, we generate a random value $ r(v)\in [0,0.5) $ and give it to this neighbor. If $ w_v<w_n $, we generate a random value $ r(v)\in (0.5,1] $ for its neighbor. When $ w_v=w_n $, we set $ r(v)=0.5 $ and sends it to this neighbor.

\item[b'.] If $ r(v)\times w_v>r(w)\times w_n $ for all neighbors $ w\in N(v) $, node $ v $ enters MIS and informs its neighbors.
\end{enumerate}

By this way, we can make the node with larger weight and fewer neighbors be added more easily to get MWIS.

\textbf{Algorithm Description} The pseudo code is shown in Algorithm 4. The algorithm operates in synchronous rounds, grouped into phases. We introduce a single phase with pseudo code. The input is an adjacent matrix $ A $ of the WCFD graph. Then we set the variable scale as the number of nodes in the graph. With the scale, we get an array $ M[scale] $ to record found maximal weight independent set (Line 1). In a single phase, for each node $ v $, we compare the weight of $ v $ with the weight of each of its neighbors. If the weight of v is larger, we generate a random number $ k $ from $ [0,0.5) $ and give it to $ w $ (Line 5-6). If the weight of $ v $ is equal to $ w $, we give $ 0.5 $ to $ w $ (Line 7-8). If the weight of $ v $ is smaller, we generate $ k $ from $ (0.5,1] $ and assign it to $ w $ (Line 9-10).

After we generate random numbers, for each node $ v $, we set a label as 1. (Line 12) We compare the random number of the neighbor of $ v $ with $ r(v) $. If the $ r(w)\times w_n $ is no smaller than $ r(v)\times w_v $, we let label be 0 (Line 15). After we finish the comparing, if label is still 1, we add $ v $ to $ M[scale] $, move $ v $ and all edges adjacent to $ v $ (Line 17-18). Then, we start another phase when there is node in $ G $ (Line 19-20).

\textbf{Time complexity Analysis} Since the modified algorithm just add the process of comparing the weight, we can use the constraints provided in \cite{Bollob1998Modern} to help analysis the time complexity. The probability in a single phase at least a quarter of all edges are removed is at least $ 1/3 $. Then with less than $ 1/3 $ for the probability , many (potentially all) edges are removed. And the probability that less than $ 1/4 $ of edges are removed is more than $ 2/3 $. Therefore, the removed edges is about $ 1/3\times 1+2/3\times 1/4=1/2 $.

Since at least $ 1/3 $ of phases are ``good'' and can remove at least a quarter of edges, we need $ \log 4/3(m) $ good phases, where the m is the amount of the edges in $ G $. The last two edges will certainly be removed in the next phase. And consider the extra time of comparing for each node, we get the $ (3\log 4/3(m)+1)\times c \in O(\log n)$ as time complexity, where c is a number no larger than the number of nodes in $ G $.

\begin{algorithm}
	\caption{MWID}
	\hspace*{0.02in}{\bf Input:}
	A graph $ \boldsymbol{G} $. The \textbf{scale} for the number of points in the graph.\\
	\hspace*{0.02in}{\bf Output:}
	The maximal weight independent set $ \boldsymbol{M[{\rm scale}]} $.
	\begin{algorithmic}[1]
		\State $ M[{\rm scale}]=\{\}; $
		\For{$ v=1 $ to scale}
			\For{each neighbor $ w $ of $ v $}
				\State \textbf{switch} {cmp$ (v,w) $}
				\State \quad \textbf{case} 1: $ k= $rand$ [0,0.5] $;
				\State \hspace{1.5cm}$ r(w)=k; $
				\State \quad \textbf{case} 2: $ k= 0.5; $
				\State \hspace{1.5cm}$ r(w)=k; $
				\State \quad \textbf{case} 3: $ k= $rand$ [0.5,1] $;
				\State \hspace{1.5cm}$ r(w)=k; $
			\EndFor
		\EndFor
		\For{$ v=1 $ to scale}
			\State label$ =1 $;
			\For{each neighbor $ w $ of $ v $}
				\If{$ r(v)\times w_v\leq r(w)\times w_n$}
					\State label$ =0 $;
				\EndIf
			\EndFor
			\If{label$ =1 $}
				\State add $ v $ to $ M[{\rm scale}] $;
				\State remove $ v $ and all edges adjacent to $ v $ from $ G $;
			\EndIf
		\EndFor
		\If{there is node in the $ G $}
			\State go to 2;
		\EndIf
	\end{algorithmic}
	\hspace*{0.02in}{\bf Note:}
	cmp$ (v,w) $ is to compare weight of $ v $ with weight of $ w $. If weight of $ v $ is bigger, it returns 1. If weight of $ v $ equals the $ w $, it returns 2. As for the condition $ v $ is smaller, we get 3.
	\label{algor: MWID}
\end{algorithm}

\section{Selection of Parameters}
\label{section: sel of para}

In CFD discovery algorithms, following parameters should be known.

\begin{enumerate}
	\item[(1)] The high limit of number of groups extracted from data set $ (n) $;
	\item[(2)] The amount of the items in each group $ (m) $;
	\item[(3)] The least number of same attributes to decide whether a tuple is similar to others $ (b) $;
	\item[(4)] The highest number of same attributes a special item has with popular items $ (b') $ and
	\item[(5)] threshold we set when find CFDs $ (e) $.
\end{enumerate}

In this section, we discuss the parameter selection methods based on user requirements. The requirements include 4 dimensions. These dimensions have trade-off. The four dimensions of CFD discovery methods are as follows.

\begin{enumerate}
	\item[(1)] \textbf{The time of finding CFDs (CW):} We always wish it can cost less time to find CFDs. The time of our algorithm is the sum time of sampling and finding CFDs from samples.
	
	\item[(2)] \textbf{The quality of CFDs (QC):} We want to improve quality of CFDs making it fit to CFDs found on clean data set. This dimension is described by percentage of CFDs found from clean data set covered by those found in dirty CFDs.
	
	\item[(3)] \textbf{The time of cleaning data with our CFDs (CC):} Another target of CFDs discovery is to clean data efficiently. We measure the time by cleaning data with CFD set.
	
	\item[(4)] \textbf{The quality of cleaning (denoted by QD):} Meanwhile, we need to ensure our CFDs clean the data effectively. We use the percentage of dirty items in data set found by CFD set to measure.
\end{enumerate}

For these parameters, a user could select a dimension as the one with the highest priority. We denote such dimension as OD. For others, the tolerate range are set. As an example, a possible demand description is as follows.

\begin{enumerate}
	\item[(1)]	We want the discovery time to be as small as possible.
	\item[(2)]	The lowest quality of CFDs we allow is 96\%.
	\item[(3)]	The longest time of using CFDs to clean our data set we allow is 3 hours.
	\item[(4)]	The lowest quality of the cleaning result is 95\%.
\end{enumerate}

We designed experimental methods to obtain these parameters according to these requirements. The data for this experiment are generated by the TPC-H. We generate a small tuple set with the same amount of attributes as those tuples in big data to be cleaned, which can make our data similar to big data ensure the functions found from our data can work well with big data.

In each experiment, we vary one parameter p1 with the others unchanged and use our method to find CFDs and clean data set by the discovered CFDs. Then we measure the amount for four aims. With some rounds of experiments, we draw a curve about the four goals and p1. By fitting such curve, we can get four functions between CW, QC, CC, QD and the parameter $ p_1 $.

With the same process, we get the functions between CW, QC, CC, QD and other parameters. We denote the relation function between $ p_i $ and CW as $ f_{cw}(p_i) $ which is similar to QC, CC and QD.

Finally, we integrate all the functions to get equations:
\begin{align*}
	{\rm CW}&=\sum_{i=1}^{r}f_{\rm CW}(p_i)/r;\\
	{\rm QC}&=\sum_{i=1}^{r}f_{\rm QC}(p_i)/r;\\
	{\rm CC}&=\sum_{i=1}^{r}f_{\rm CC}(p_i)/r;\\
	{\rm QD}&=\sum_{i=1}^{r}f_{\rm QD}(p_i)/r.
\end{align*}

Then, we can formalize the problem as an optimization problem with description of one of CW, QC, CC and QD as optimization goal, and other equations as well as the input range requirements as the constraint. By applying Simplex Algorithm, we get the optimized solution for this problem to determine the parameters. For example, in our experiment it is solved as following.

Then, we can get
\begin{align*}
	&{\rm QC}=\sum_{i=1}^{r}f_{\rm QC}(p_i)/r;\\
	&{\rm CC}=\sum_{i=1}^{r}f_{\rm CC}(p_i)/r;\\
	&{\rm QD}=\sum_{i=1}^{r}f_{\rm QD}(p_i)/r;\\
	&n\times m<N/50000;\\
	&50\%<e<1;\\
	&n\geq 2;\\
	&2<b'<b<15;\\
	&QC\geq 0.95;\\
	&CC\leq 130;\\
	&QD\geq 0.95;\\
	&CW\leq 130.
\end{align*}

Using the Simplex Algorithm, we get parameter set $ \{n=11, m=4000, e=0.9, b'=4, b=9\} $ which is proved as a best choose in Section 7.2.2.

\begin{figure*}[h]
	\centering
	\subfloat[\#tuple VS. response time]{
		\includegraphics[width=.3\textwidth]{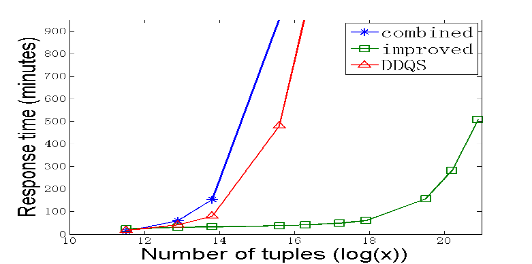}}\hfill
	\subfloat[\#attribute VS. response time]{
		\includegraphics[width=.3\textwidth]{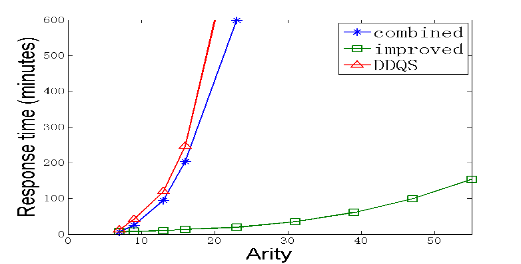}}\hfill
	\subfloat[\#tuple VS. inconsistent rate]{
		\includegraphics[width=.3\textwidth]{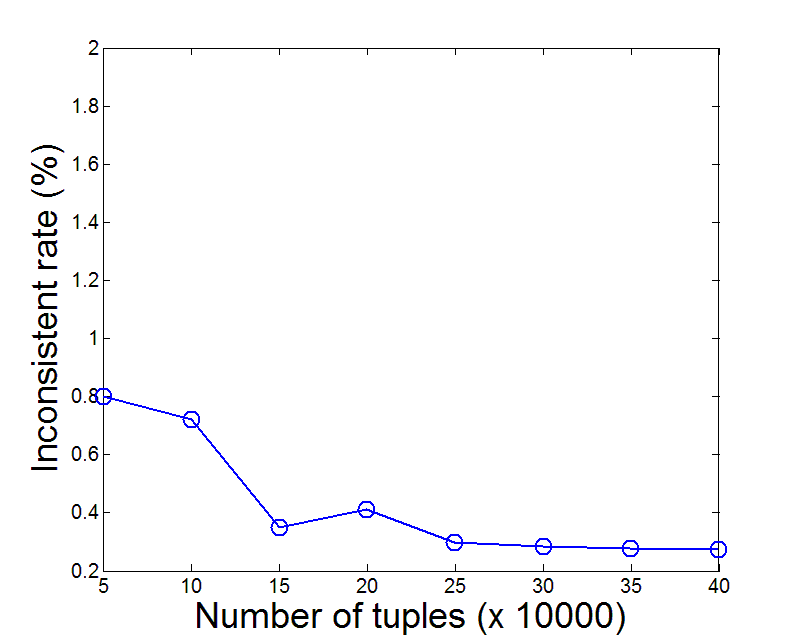}}\hfill\\
	\subfloat[\#attribute VS. inconsistent rate]{
		\includegraphics[width=.3\textwidth]{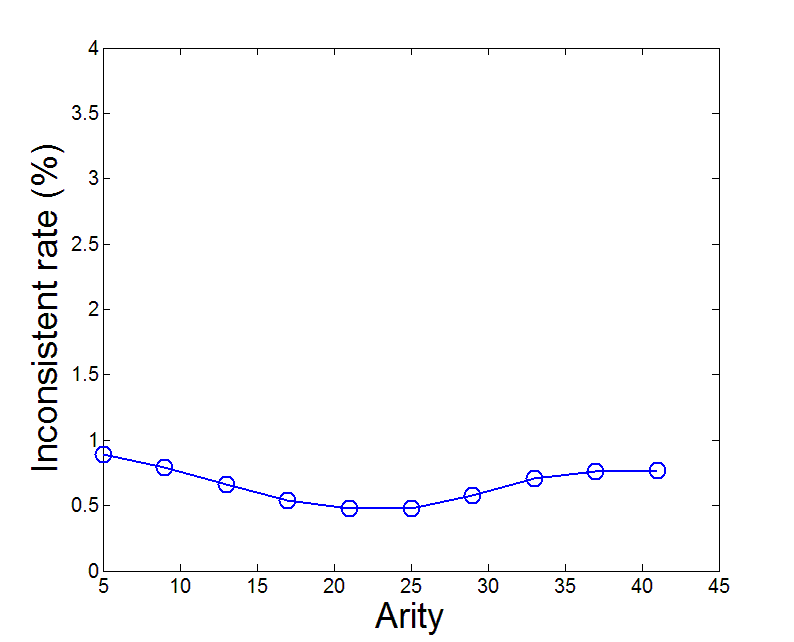}}\hfill
	\subfloat[The effect of our method on SUSY]{
		\includegraphics[width=.3\textwidth]{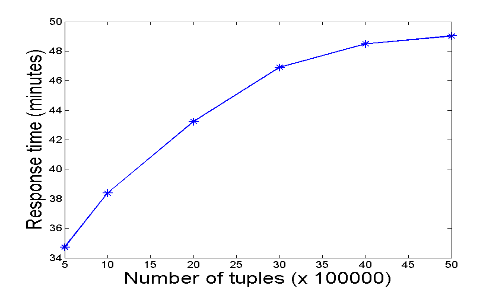}}\hfill
	\subfloat[The effect of our method on Article]{
	\includegraphics[width=.3\textwidth]{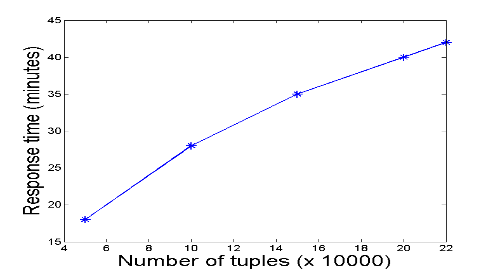}}\hfill
	\caption{Experiment result}
	\label{fig: exp result}
\end{figure*}

\section{Experiments}
\label{section: exp}

To verify the efficiency and effectiveness of proposed algorithms, we perform extensive experiments in this section.

\subsection{Experimental Settings}
The experiments were conducted on both synthetic data sets and real-life data. We firstly use synthetic data generated by TPC-H, which is a decision support benchmark and can generate data in any size to evaluate performance and scalability of our algorithm and optimality of the method of choosing parameters. We also used real data sets name from the UCI machine learning repository\footnote{\url{http://archive.ics.uci.edu/ml/}} and dblp\footnote{\url{http://dblp.uni-trier.de/}} namely SUSY Data Set, Article Data Set showed in Table \ref{tab: para of real dataset} to check effect on real data of the method.

\begin{table}
	\centering
	\caption{The Parameters of Real Dataset}
	\label{tab: para of real dataset}
	\begin{tabular}{|c|c|c|}
		\hline
		Dataset & Arity & size(\# of tuples) \\
		\hline
		SUSY Data Set & 18 & 5,000,000 \\
		\hline
		Article Data Set & 23 & 220,000 \\
		\hline
	\end{tabular}
\end{table}

All algorithms are implemented in Java. The program has been tested on a PC with Intel kurui i7 4770 (3.4GHZ) and 8 GB of memory running Ubuntu operating system. Each experiment was repeated three times and the average is reported.

We use following parameters to evaluate proposed algorithms.

\begin{enumerate}
	\item[(1)]	the time of finding CFDs from the data set.
	\item[(2)]	the quality of discovered CFDs measured by the percentage of the standard CFD set got from clean data found by us using our method on dirty data changed from the clean data.
	\item[(3)]	the time of cleaning data with discovered CFDs.
	\item[(4)]	the quality of data cleaned by discovered CFDs measured by using CFDs to clean dirty data changed from clean data and measure the percentage of dirty tuples found by our CFDs
\end{enumerate}

Also, to test the optimality of the method we choose parameters in Section 6, we compare the effect of different choice of parameters using the controlling variable method.

\subsection{Experiment Result}
\subsubsection{Performance and scalability experiments}
We show performance and scalability of our algorithm by different data size and arity. We use CFDs discovered on clean data by combined algorithm (the original CFD discovery algorithms) as baseline. And we modify 8\% of generated data to make it dirty and use the dirty data to test.

\noindent\textbf{A.	Efficiency Experiments}

\noindent\textbf{a.	The Impact of Tuple Number}

We varied tuple number from 100K to 1.2 Billion with 16 attributes for each tuple. The maximal data size is 210G. The discovery time is shown in Figure \ref{fig: exp result}(a), where ¡°combined¡± refers to the original algorithms for CFD discovery and the ¡°improved¡± refers to the algorithm proposed in this paper and DDQS as the experiment for \cite{Fei2008Discovering}. The horizontal axis is in logrithm scale. The reason we do not use x directly is that other two algorithms can only work for small data and we want to use very big size to show data size our algorithm can deal with. From this figure, we have following observations.

\begin{enumerate}
	\item[(1)] When DBSIZE (the size of database) is smaller than 70K, the response time of our method is higher than that of combined and DDQS algorithm. This shows that due to time of sampling and conflict resolution, our method behave bad with small data.
	\item[(2)] When DBSIZE$ > $70K, original algorithms and DDQS find CFDs more slowly. It is because that when DBSIZE is big enough, it will cost more time to find CFDs than sampling and combining different sets.
	\item[(3)] The increasing speed of other two lines is higher than ours. This shows that our algorithm works better for big data.
\end{enumerate}

\noindent\textbf{b. The Impact of Attribute Number}

We vary the attribute number from 7 to 55 and fix the tuple number 1,000K. From Figure \ref{fig: exp result}(b), we can find the two lines are both index functions which shows the index form of $ r $ in the objective function for O1. However, the line for our method is more gentle. When arity is smaller than 23, combined algorithms is faster because there is no need to sample. When it comes to the arity of more than 25 attributes, our method outperforms combined algorithm. Our method can save over 20 percentage of the time when arity is 55. Compared with other lines in graph, we can forecast our method does better for the data with more attributes.

\begin{table}
	\centering
	\caption{The Optimization of $ n $}
	\label{tab: opti of n}
	\begin{tabular}{|c|c|c|c|c|c|}
		\hline
		$ n $ & 5 & 8 & 11 & 14 & 17 \\
		\hline
		CW & 75min & 83min & 92min & 123min & 157min \\
		\hline
		QC & 0.972 & 0.99 & 0.994 & 0.994 & 0.993 \\
		\hline
		CC & 102min & 115min & 123min & 132min & 139min\\
		\hline
		QD & 0.83 & 0.94 & 0.96 & 0.97 & 0.98\\
		\hline
	\end{tabular}
\end{table}

\begin{table}
	\centering
	\caption{The Optimization of $ m $}
	\label{tab: opti of m}
	\begin{tabular}{|c|c|c|c|c|c|}
		\hline
		$ m $ & 2,000 & 4,000 & 6,000 & 8,000 & 10,000 \\
		\hline
		CW & 99min & 102min & 121min & 155min & 190min \\
		\hline
		QC & 0.985 & 0.994 & 0.995 & 0.995 & 0.997 \\
		\hline
		CC & 95min & 123min & 132min & 141min & 151min\\
		\hline
		QD & 0.92 & 0.96 & 0.97 & 0.97 & 0.98\\
		\hline
	\end{tabular}
\end{table}

\begin{table}
	\centering
	\caption{The Optimization of $ e $}
	\label{tab: opti of e}
	\begin{tabular}{|c|c|c|c|c|c|}
		\hline
		$ e $ & 0.6 & 0.7 & 0.8 & 0.9 & 0.97 \\
		\hline
		CW & 120min & 107min & 99min & 92min & 86min \\
		\hline
		QC & 0.951 & 0.979 & 0.987 & 0.994 & 0.985 \\
		\hline
		CC & 150min & 138min & 129min & 123min & 120min\\
		\hline
		QD & 0.86 & 0.89 & 0.93 & 0.96 & 0.94\\
		\hline
	\end{tabular}
\end{table}

\begin{table}
	\centering
	\caption{The Optimization of $ b $}
	\label{tab: opti of b}
	\begin{tabular}{|c|c|c|c|c|}
		\hline
		$ b $ & 7 & 9 & 10 & 12 \\
		\hline
		CW & 81min & 92min & 103min & 124min\\
		\hline
		QC & 0.99 & 0.994 & 0.993  & 0.991 \\
		\hline
		CC & 102min & 123min & 128min & 131min\\
		\hline
		QD & 0.94 & 0.96 & 0.95 & 0.94\\
		\hline
	\end{tabular}
\end{table}

\begin{table}
	\centering
	\caption{The Optimization of $ b' $}
	\label{tab: opti of b'}
	\begin{tabular}{|c|c|c|c|c|}
		\hline
		$ b' $ & 2 & 4 & 5 & 6 \\
		\hline
		CW & 95min & 92min & 101min & 104min\\
		\hline
		QC & 0.972 & 0.994 & 0.99  & 0.983 \\
		\hline
		CC & 147min & 123min & 117min & 109min\\
		\hline
		QD & 0.96 & 0.96 & 0.95 & 0.93\\
		\hline
	\end{tabular}
\end{table}

\noindent\textbf{B.	Precision Experiments}

We add the CFDs found by the methods in the Gather of Original Algorithms (discover CFDs using the original algorithms) to get a standard set of CFDs. By computing the percentage of the standard set of CFDs which are not covered by our CFDs, we evaluate the precision of our algorithm.

\noindent\textbf{a.	The Impact of Tuple Number}

We varied the tuple number from 50K to 400K tuples with 16 attributes for each item. The line shows the percentage of the CFDs found by the methods in the combined algorithm. The y-axis represents the percentage of CFDs, which are generated by the combined algorithm, covered by CFDs from our algorithm and is called inconsistent rate.

From Figure \ref{fig: exp result}(c), we have following observations. (1) The inconsistent rate is less than 1\%. This shows our method can always find CFDs which are similar to those found by original methods. (2) When tuples are large, inconsistent rate is smaller than those for small number of tuples. This proves our method is more scalable for big data. (3) When DBSIZE is larger than 30K, we can use CFDs we find as a standard set of CFDs due to precision $ >96.5\% $.

\noindent\textbf{b.	The Impact of Attribute Number}

We varied the arity from 5 to 42 by fixing the tuples as 100K. From Figure \ref{fig: exp result}(d), we can find that the CFDs found by our method are very similar to those standard CFDs no matter what the arity is. When arity is 25, effect is the best of all. However, when the arity is too big or too small, the result of our CFDs becomes worse. As for the too big arity, wrong items may be concentrated when we change items by ourselves. And we will obtain some similar wrong samples, which makes our CFDs seems wrong. This is caused by people and to real data set, this will not happen. For too small arity, we can find CFDs are few making base small.

\subsubsection{Optimality of Parameters}

To show the effectiveness of our parameter selection method, we change one parameter with others unchanged and compare the four dimensions.  We use TPC-H to generate data set and use the default constraint for parameters as $ \{n=11, m=4000, e=0.9, b'=4, b=9\} $ and only change a parameter in these parameters.  In each set of experiments, we compare the results with parameters $ \{n=11,m=4000, e=0.9, b'=4, b=9\} $ obtained by our method and those with different values of the optimization goal. In each table, column in grey background is the result with optimal parameters.

The results with $ n $ as the optimization goal is shown in Table \ref{tab: opti of n}. From the results, we find that only when $ n $ is 11, the four aims can be satisfied and the CW is as short as impossible.  When it is large, the time for finding CFDs will increase and when it is too small, the CFDs we find will be not so accurate. From Table \ref{tab: opti of m}, we can see that the optimization result for m is 4,000 items in each group of sampling. When it is too large, we can find that the time of finding CFDs and cleaning data set is too large. When it is too small, QC will be worse and result of cleaning is bad.

For the $ e $, we know from Table \ref{tab: opti of e} is 0.9, we can get the best results. When $ e $ is too small, we will find many wrong CFDs and spend a lot of time to clean data. When $ e $ is too large, we can be too strict to tolerant wrong tuples and leave CFDs.

The results with $ b $ as the optimization goal is shown in Table \ref{tab: opti of b}. We can see that when it is too small or too large, the CFDs we find will be not so accurate.

To the $ b' $, we can find from Table \ref{tab: opti of b'} that we should choose 4 as its amount. Although when it is 5, we can also accept the result, while the CW being smaller than 4 make us abandon it.

\textbf{Conclusion:} Above all, we can see that our selection of {$ n=11,m=4000, e=0.9, b'=4, b=9 $} can work best to make CW as small as possible and satisfy the low limits for other aims.

\subsubsection{Test on Real Data}
We use real data sets from the UCI machine learning repository and dblp namely SUSY Data Set, Article Data to test the effectiveness of our method on real data.

Figure \ref{fig: exp result}(e) shows time of discovering CFDs from SUSY Data Set. We use different part of the data set to test the scalability. We observe that when we increase the tuple number, time increases around linear with the data size. It shows that our method can deal with real big data in a linear effect. Largest size of data is 6.2G.To data set from dblp in Figure \ref{fig: exp result}(f), we vary DBSIZE from 50K to 200K. The largest size of data is 2.1G. We also find that the time of finding CFDs increases linear with the data size. This also shows the linear cleaning effect of our method on real big data. Thus, the experimental results on real data verify analysis results.

\section{Related Work}

In this section, we give a brief survey of related work.

\noindent\underline{Concept of Dependencies} To improve data consistency, a set of data quality rules are often created. Once the inconsistent items exist in the database, some rule will be violated. Thus, errors are discovered and revised correspondingly. It has been generally admitted that the integrity constraints should be used as a data quality detection rule to improve data consistency \cite{Fan2011Discovering}\cite{Fei2008Discovering}\cite{Fan2008Conditional}\cite{Cormode2009Estimating}\cite{Batini2006Data}.

The theory of conditional dependencies including CFDs \cite{Chomicki2006Consistent} and conditional inclusion dependencies (CINDs) \cite{Bertossi2006Consistent} respectively develops the traditional functional dependencies and inclusion dependencies to capture the common mistakes in realistic data.

For the conditional functional dependencies, \cite{Chomicki2006Consistent} and \cite{Bertossi2006Consistent} study the problems including the consistency, logical implication and axiomatic for dependency language. Based on the \cite{Chomicki2006Consistent} and \cite{Bertossi2006Consistent}, a variety of extensions for conditional dependencies have been proposed in \cite{Fan2008Dependencies}, \cite{Rahm2010Data}, \cite{Fan2008Conditional} and \cite{Bravo2007Extending} to develop the capacity of illustrating conditional dependencies without the growth of the computational complexity.

\noindent\textbf{(2) Rules Mining}

To use dependencies as data quality rules, the first problem is how to obtain these dependencies. \cite{Bravo2008Increasing} and \cite{Chen2009Analyses} design the automatic discovery algorithms for finding CFDs. However, the algorithms in them both need to work on a clean and representative data set. In \cite{Bollob1998Modern}, it can discover CFDs from dirty data set. However, the process can be hardly finished for the data set with size larger than memory. Meanwhile, the complexity of algorithm in \cite{Bollob1998Modern} is large for big data. For CINDs, no automatic discovery algorithm has been proposed. In [23], it can find CINDs based on dirty data set. What is more, [21] can discover the Integrity Constraints (ICs) for data set which is not clean automatically. And both [21] and [23] do not consider the problem of memory size.

\noindent\textbf{(3) Algorithms used in rules mining for big data}

To find rules on big data, people give many methods. In [22], it proposes an ¡°on-demand¡± algorithm to generate an optimal tableau for given CFDs. In \cite{Vitter1985Random}, it uses a variety of sampling and sketching techniques to estimate the confidence of a CFD with a small number of passes (one or two) over input using small space.

\noindent\textbf{(4) Rules Analysis and Optimization}

Since data quality rule set may contain conflicts, we need to find out consistent constraint rules (i.e., maximum consistent subset) as data quality rules. The computational complexity of this problem is very high. For CFDs, finding the maximum consistent subset of rules is proven to be NP-complete~\cite{Kalavagattu2008MINING}. When we consider both the CFDs and CINDs, this problem is undecidable. Thus, approximate algorithms to calculate maximum consistent subset for CFDs have been proposed in \cite{Chomicki2006Consistent}.

\noindent\textbf{(5) Error Detection}

Error detection means capturing data errors by the consistent subset of the data quality rules. It finds the tuples in violation to the data quality rules. \cite{Chomicki2006Consistent} and \cite{Rahm2010Data}, for centralized storing relational databases, are designed to detect the tuples in violation of CFDs and CINDs automatically based on SQL query processing.

\section{Conclusions}
\label{section: con}

For big data, rule discovery for data cleaning brings new challenges. To solve this problem, we propose a novel CFD discovery method for big data. For the volume feature of big data, we design a sampling algorithm to obtain typical samples by scanning data only once. Then, on the sample set, we adapt existing CFD discovery algorithms to tolerate the fault. By integrating these modified methods, a preliminary CFD set is discovered. To increase the quality in the discovered rule set, we design a graph-based rule selection algorithm. With the consideration that a user may have different requirements for CFD discovery, we propose a strategy to select parameters according to requirements of users. Experimental results demonstrate that proposed algorithm is suitable for big data and outperforms existing algorithms. Future work includes extending the proposed algorithm to parallel platform and modifying proposed algorithm to discover other rules.

\bibliographystyle{abbrv}
\bibliography{paper}

\begin{thebibliography}{10}

\bibitem{Batini2006Data}
C.~Batini and M.~Scannapieco.
\newblock Data quality: Concepts, methods and techniques.
\newblock 6, 2006.

\bibitem{Bertossi2006Consistent}
L.~Bertossi.
\newblock Consistent query answering in databases.
\newblock {\em Acm Sigmod Record}, 35(2):68--76, 2006.

\bibitem{Bollob1998Modern}
B.~Bollob'As.
\newblock Modern graph theory : Springer-verlag.
\newblock 1998.

\bibitem{Bravo2008Increasing}
L.~Bravo, W.~Fan, F.~Geerts, and S.~Ma.
\newblock Increasing the expressivity of conditional functional dependencies
  without extra complexity.
\newblock In {\em IEEE International Conference on Data Engineering}, pages
  516--525, 2008.

\bibitem{Bravo2007Extending}
L.~Bravo, W.~Fan, and S.~Ma.
\newblock Extending dependencies with conditions.
\newblock In {\em International Conference on Very Large Data Bases}, pages
  243--254, 2007.

\bibitem{Chen2009Analyses}
W.~Chen, W.~Fan, and S.~Ma.
\newblock Analyses and validation of conditional dependencies with built-in
  predicates.
\newblock In {\em International Conference on Database and Expert Systems
  Applications}, pages 576--591, 2009.

\bibitem{Chomicki2006Consistent}
J.~Chomicki.
\newblock Consistent query answering: Five easy pieces.
\newblock In {\em Database Theory - ICDT 2007, International Conference,
  Barcelona, Spain, January 10-12, 2007, Proceedings}, pages 1--17, 2006.

\bibitem{Cormode2009Estimating}
G.~Cormode, L.~Golab, K.~Flip, A.~Mcgregor, D.~Srivastava, and X.~Zhang.
\newblock Estimating the confidence of conditional functional dependencies.
\newblock pages 469--482, 2009.

\bibitem{Fan2008Dependencies}
W.~Fan.
\newblock Dependencies revisited for improving data quality.
\newblock In {\em Twenty-Seventh ACM Sigmod-Sigact-Sigart Symposium on
  Principles of Database Systems}, pages 159--170, 2008.

\bibitem{Fan2008Conditional}
W.~Fan, F.~Geerts, X.~Jia, and A.~Kementsietsidis.
\newblock Conditional functional dependencies for capturing data
  inconsistencies.
\newblock {\em Acm Transactions on Database Systems}, 33(2):1--48, 2008.

\bibitem{Fan2011Discovering}
W.~Fan, F.~Geerts, J.~Li, and M.~Xiong.
\newblock Discovering conditional functional dependencies.
\newblock {\em IEEE Transactions on Knowledge \& Data Engineering},
  23(5):683--698, 2011.

\bibitem{Fei2008Discovering}
C.~Fei and R.~J. Miller.
\newblock Discovering data quality rules.
\newblock {\em Proceedings of the Vldb Endowment}, 1(1):1166--1177, 2008.

\bibitem{Garey1986Computers}
M.~R. Garey and D.~S. Johnson.
\newblock Computers and intractability: A guide to the theory of
  np-completeness.
\newblock 1986.

\bibitem{Golab2008On}
L.~Golab, H.~Karloff, F.~Korn, D.~Srivastava, and B.~Yu.
\newblock On generating near-optimal tableaux for conditional functional
  dependencies.
\newblock {\em Proceedings of the Vldb Endowment}, 1(1):376--390, 2008.

\bibitem{Kalavagattu2008MINING}
A.~K. Kalavagattu.
\newblock Mining approximate functional dependencies as condensed
  representations of association rules.
\newblock 2008.

\bibitem{Rahm2010Data}
E.~Rahm and H.~D. Hong.
\newblock Data cleaning: Problems and current approaches.
\newblock {\em IEEE Data Engineering Bulletin}, 23(23):3--13, 2010.

\bibitem{Vitter1985Random}
J.~S. Vitter.
\newblock Random sampling with a reservoir.
\newblock {\em Acm Transactions on Mathematical Software}, 11(1):37--57, 1985.

\end{thebibliography}

\end{spacing}{1.0}
\end{document}